\documentclass[preprint,12pt,authoryear]{elsarticle}

\pdfoutput=1
\usepackage{amssymb}
\usepackage{amsthm}
\usepackage{multirow}%
\usepackage{amsmath,amssymb,amsfonts}%
\usepackage{amsthm}%
\usepackage{mathrsfs}%
\usepackage[mathscr]{eucal}
\usepackage[title]{appendix}%
\usepackage{xcolor}%
\usepackage{textcomp}%
\usepackage{manyfoot}%
\usepackage{booktabs}%
\usepackage{algorithm}%
\usepackage{algorithmicx}%
\usepackage{algpseudocode}%
\usepackage{listings}%
\usepackage{subfigure}%
\usepackage[titles,subfigure]{tocloft}
\usepackage{caption}
\usepackage{hyperref}
\usepackage{indentfirst}
\usepackage{soul} 
\usepackage{color, xcolor} 
\usepackage{epsfig}
\usepackage{graphicx}
\usepackage{boondox-cal}
\usepackage{manyfoot}%

\soulregister{\cite}7
\soulregister{\citep}7
\soulregister{\citet}7
\soulregister{\ref}7
\soulregister{\pageref}7

\usepackage{hyperref}
\hypersetup{
hidelinks,
colorlinks=true,
citecolor=blue}

\def\mathcal{\mathscr}
\newtheorem{lemma}{Lemma}[]

\theoremstyle{thmstyleone}%
\newtheorem{theorem}{Theorem}

\theoremstyle{thmstyletwo}%

\theoremstyle{thmstylethree}%
\journal{Applied Soft Computing}

\begin{document}

\begin{frontmatter}




\title{Robust kernel-free quadratic surface twin support vector machine with capped $L_1$-norm distance metric}

\author[1,2]{Qi Si}
\author[1,2]{Zhixia Yang\corref{cor1}}\ead{yangzhx@xju.edu.cn}
\cortext[cor1]{Corresponding authors}
\address[1]{College of Mathematics and Systems Science, Xinjiang University, Urumuqi 830046, China}
\address[2]{Institute of Mathematics and Physics, Xinjiang University, Urumuqi 830046, China}

\begin{abstract}
Twin support vector machine (TSVM) is a very classical and practical classifier for pattern classification. However, the traditional TSVM has two limitations. Firstly, it uses the $L_2$-norm distance metric that leads to its sensitivity to outliers. Second, it needs to select the appropriate kernel function and the kernel parameters for nonlinear classification. To effectively avoid these two problems, this paper proposes a robust capped $L_1$-norm kernel-free quadratic surface twin support vector machine (C$L_1$QTSVM). The strengths of our model are briefly summarized as follows. 1) The robustness of our model is further improved by employing the capped $L_1$ norm distance metric. 2) Our model is a kernel-free method that avoids the time-consuming process of selecting appropriate kernel functions and kernel parameters. 3) The introduction of $L_2$-norm regularization term to improve the generalization ability of the model. 4) To efficiently solve the proposed model, an iterative algorithm is developed. 5) The convergence, time complexity and existence of locally optimal solutions of the developed algorithms are further discussed. Numerical experiments on numerous types of datasets validate the classification performance and robustness of the proposed model.
\end{abstract}


\end{frontmatter}


\section{Introduction}
  With the development of data science, the processing and learning of data is increasingly important. In view of the strong generalization ability of the support vector machine (SVM) \citep{SVM} and its following the principle of structural risk minimization (SRM), which leads to its increasing application in classification and regression learning tasks \citep{LpLSTSVM,Huber-huigui-ELM}. In addition, the classical SVM deal with nonlinear classification problems by introducing suitable kernel functions \citep{hejishu}. Nowadays, SVM models have been widely applied in the fields of fault diagnosis \citep{guzhang}, pattern classification \citep{LpLSTSVM}, gene expression \citep{jiyin}, face recognition \citep{yingyong3}, and text classification \citep{yingyong4}.

  However, the classical SVM model can be limited to handle certain data distributions, such as the XOR problem. In addition, traditional SVMs obtain the optimal solutions of the primal problem by solving a quadratic programming problem (QPP), which leads to some limitations in dealing with large-scale problems. To address the above two problems, twin SVM (TSVM) has been proposed \citep{TSVM}.  In comparison to classical SVMs, traditional TSVM solve two smaller size quadratic programming problems (QPPs), which leads to its computational speed theoretically to be faster than that of SVMs. The classical TSVM only considers the empirical risk minimization principle and not the SRM principle. Therefore, to make TSVM follow the SRM principle, Shao et al. proposed twin bounded SVM (TBSVM) by introducing a regularization term in the TSVM model \citep{zhengzexiang}.  In recent years, the TSVM model has been further improved and applied to the fields of robust classification learning \citep{SLTSVM}, multi-view data classification \citep{duoshijiao1}, unbalanced data classification \citep{TSVM2} and deep learning \citep{xieshenduTSVM}.

  Recently, some scholars have proposed numerous improvements by combining kernel-free techniques with SVM and TSVM models.  The kernel-free soft-margin quadratic surface support vector machine (SQSSVM) was first proposed by Luo {\it et al.} for the binary classification problem \citep{luoSQSSVM}. Similarly, Bai {\it et al.} proposed quadratic least squares support vector machine (QLSSVM) \citep{baiQSLSSVM}. Subsequently, Gao et al. proposed quadratic TSVM (QTSVM) and quadratic least squares TSVM model (LSQTSVM) by combining kernel-free techniques with TSVM \citep{gaoQSLSTSVM}. In order to make QSSVM deal with high-dimensional feature problems effectively, Gao {\it et al.} proposed $\nu$-version fuzzy reduced SQSSVM ($\nu$-FRSQSSVM) \citep{gaonu-FRSQSSVM}. In recent years, the double well potential (DWP) function, which is a 4th order polynomial function, has also been studied more and more extensively\citep{shuangshijin1,shuangshijin2,gao-DWPSVM,gao-reg-LSDWPTSVM}.  Although the above kernel-free methods improve the performance of SVM and TSVM models for nonlinear classification problems, they have not been applied to the field of robust classification learning, and thus this is an extendable research direction.

   Although TSVM has achieved good performance in pattern classification, there is still room for improvement. The main reason for this is that TSVM uses the $L_2$ norm distance metric and the hinge loss function, which causes it to be sensitive to outliers or noises \citep{L2yichang,Hingeyichang}. In order to enhance the robustness of the model, some scholars have proposed many improvements based on the distance norm. For instance, Wang {\it et al.} combined the capped $L_1$-norm distance metric with TSVM (CTSVM) to deal with robust classification problems \citep{cappedL1TSVM}. In order to make CTSVM follow the SRM principle, Ma {\it et al.} proposed capped $L_1$-norm distance metric-based fast robust TBSVM (FRTBSVM) by introducing a regularization term in the CTSVM model \citep{FRTBSVM}. Similarly, Yuan {\it et al.} introduced the new capped $L_{2,p}$ norm metric in the least squares TSVM model \citep{CL2p}. Recently, Wang {\it et al.} also introduced capped $L_{2,p}$ norm metric and bounded Welsch loss function into TSVM \citep{CL2p1}.  In addition, the capped $L_1$-norm distance metric has been introduced into methods such as the twin extreme learning machines (C$L_1$TELM) \citep{L1TELM} and the twin projection extreme learning machines (C$L_1$TPELM) \citep{L1TPELM}. It can be seen that the above methods improve the robustness of TSVM, but they are not combined with kernel-free techniques, which is very interesting.

   Noting that no scholars have combined capped $L_1$-norm distance metric with kernel-free TSVM to solve robust classification problems until now, this paper proposes a new kernel-free quadratic surface TSVM model based on capped $L_1$-norm distance metric. Our model is to find two quadratic surfaces such that each surface is as close to one class of points and as far away from the other class of points as possible. The main contributions of the paper are summarized as follows:

\begin{enumerate}
  \item[$\bullet$] The robustness of the LSQTSVM model to outliers is improved by introducing capped $L_1$-norm loss.
  \item[$\bullet$]  To comply with the structural risk minimization (SRM) principle, we introduce a regularization term to improve the generalization ability of model.
      \item[$\bullet$] And efficient and fast iterative algorithms are developed to solve the optimization problem of proposed models.
        \item[$\bullet$] The convergence, time complexity and existence of optimal solutions of our algorithm are analyzed theoretically to further validate the effectiveness of our C$L_1$QTSVM model.
         \item[$\bullet$] Experiments on synthetic and real datasets validate that the proposed C$L_1$QTSVM model slightly outperforms the compared state-of-the-art methods in terms of robustness.
    \end{enumerate}

The remaining parts of this paper are summarized as follows. Section \ref{2} briefly reviews the work related to this paper, including some definitions, models and norms. Section \ref{3} presents our model and its corresponding solution algorithm. Furthermore, some further theoretical analysis of the algorithm is provided. Section \ref{4} shows the results of numerical experiments on synthetic and real datasets. Section \ref{6} briefly summarizes the contributions of this paper and future work.


\section{Related work}\label{2}
			In this section, we first introduce some notations and double well potential functions used in the paper. Then, some related models and  distance metrics are briefly reviewed.
 \subsection{Notation}
 Throughout this paper, scalars, vectors and matrices are represented in italics. Furthermore, scalars are not required to be bolded, while vectors and matrices are required to be bolded. $\mathbb{R}^{n}$ and $\mathbb{R}^{n}_{+}$ denote the $n$-dimensional real number space and the $n$-dimensional non-negative real space respectively. $\mathbb{S}^{n}$ and $\mathbb{D}^{n}$ represent the real symmetric and real diagonal matrices of $n\times n$, respectively. The $\boldsymbol{0}$ means a zero matrix of appropriate dimension, while $\boldsymbol{I}$ represents a identity matrix of arbitrary dimension. All vectors used in this paper default to column vectors. $\top$ indicates a transpose operation. $\boldsymbol{e}_+$ and $\boldsymbol{e}_-$ represent $m_+$ and $m_-$ dimensional all-one column vectors, respectively.

For a binary dataset $\mathcal{T}$, it can be formulated mathematically as follows
\begin{equation}\label{shujuD}
\mathcal{T}=\left\{\left(\boldsymbol{x}_{i}^{\pm}, y_{i}^{\pm}\right)_{i=1,2,...,m_{\pm}}|\boldsymbol{x}_{i}^{\pm}\in \mathbb{R}^{n}, y_{i}^{\pm} \in \left\{+1,-1\right\}\right\},
\end{equation}
 where $m = m_++m_-$, $m_+$ and $m_-$ represent the number of positive and negative sample points, respectively. Furthermore, $n$ represents the feature dimension of the data $\boldsymbol{x}$. Meanwhile, the positive and negative sample matrices are defined as $\boldsymbol X_{+}=\left[\boldsymbol x_{1}^{+}, \ldots, \boldsymbol x_{m_{+}}^{+}\right]^{\top} \in \mathbb{R}^{m_{+} \times n}$ and $\boldsymbol X_{-}=\left[\boldsymbol x_{1}^{-}, \ldots, \boldsymbol x_{m_{-}}^{-}\right]^{\top} \in \mathbb{R}^{m_{-} \times n}$, respectively.

Then, we briefly introduce the four vectorization operators employed in this paper, which are extensively used in kernel-free SVM , TSVM models  \citep{gaonu-FRSQSSVM,gao-reg-LSDWPTSVM}. For any $\boldsymbol{A}\in\mathbb{S}^{n}$, its semi-vectorized operator follows
\begin{equation}\label{dingyi2}
  \begin{split}
 \operatorname{hvec}(\boldsymbol{A}) & \triangleq  {\left[A_{11}, \ldots, A_{1 n}, A_{22}, \ldots, A_{2 n}, \ldots, A_{n-1, n-1}, A_{n-1, n}, A_{n n}\right]^{\top} } \\
 & \in \mathbb{R}^{n(n+1) / 2}.
  \end{split}
\end{equation}

If $\boldsymbol{A}\in\mathbb{D}^{n}$, the diagonal vectorization operator is
\begin{equation}\label{dingyi3}
\operatorname{dvec}(\boldsymbol{A}) \triangleq \left[A_{11}, A_{22}, \ldots, A_{n-1, n-1}, A_{n n}\right]^{\top}  \in \mathbb{R}^{n}.
\end{equation}

Furthermore, given the vector $\boldsymbol{x}=\left[x_1, \ldots, x_n\right]^{\top}\in\mathbb{R}^{n}$, its quadratic vectorization operator is defined as follows
\begin{equation}\label{dingyi4}
\operatorname{lvec}(\boldsymbol{x}) \triangleq \left[\frac{1}{2}x_{1}^{2}, x_{1}x_{2}, \ldots, x_{1}x_{n}, \frac{1}{2}x_{2}^{2}, x_{2}x_{3}, \ldots, x_{2}x_{n}, \ldots, \frac{1}{2}x_{n}^{2}\right]^{\top}  \in \mathbb{R}^{\frac{n(n+1)}{2}},
\end{equation}
and its quadratic vectorization operator without cross terms can be defined as follows
\begin{equation}\label{dingyi5}
\operatorname{qvec}(\boldsymbol{x}) \triangleq \left[\frac{1}{2}x_{1}^{2}, \frac{1}{2}x_{2}^{2}, \ldots, \frac{1}{2}x_{n}^{2}\right]^{\top}  \in \mathbb{R}^{n}.
\end{equation}

		\subsection{Related models}
In this subsection, we will briefly review two classical correlation methods called LSQTSVM and FRTBSVM.

        \subsubsection{LSQTSVM}
        Recently, kernel-free TSVM methods have been further extensively studied. Gao {\it et al.} proposed the least squares quadratic TSVM (LSQTSVM) model \citep{gaoQSLSTSVM}. In particular, the LSQTSVM model solves a pair of quadratic surfaces as follows
        \begin{equation}\label{erciqumian}
\frac{1}{2} \boldsymbol{x}^{\top} \boldsymbol{A}_{\pm} \boldsymbol{x}+\boldsymbol{b}_{\pm}^{\top}\boldsymbol{x}+c_{\pm}=0,
\end{equation}
where $\boldsymbol{A}_{\pm}\in\mathbb{S}^{n}$, $\boldsymbol{b}_{\pm}\in\mathbb{R}^{n}$, $c_{\pm}\in\mathbb{R}$.

To obtain the above pair of quadratic surfaces (\ref{erciqumian}), the LSQTSVM model is built as follows

$\text{(LSQTSVM1)}$
\begin{equation}\label{LSQTSVM1}
	\begin{split}
\ \ \ \qquad \min \ & \frac{1}{2}\sum_{i=1}^{m_{+}}(\frac{1}{2} {\boldsymbol{x}_{i}^{+}}^{\top} \boldsymbol{A}_+ \boldsymbol{x}_{i}^{+}+ \boldsymbol{b}_+^{\top}\boldsymbol{x}_{i}^{+}+c_+)^2 +C \sum_{j=1}^{m_{-}} (\xi_{j}^{-})^2, \\	
\ \ \ \qquad \text{s.t.} \ & -(\frac{1}{2} {\boldsymbol{x}_{j}^{-}}^{\top} \boldsymbol{A}_+ \boldsymbol{x}_{j}^{-}+\boldsymbol{b}_+^{\top}\boldsymbol{x}_{j}^{-}+c_+)=1-\xi_{j}^{-},\quad j=1,2, \ldots, m_{-},\\
& \boldsymbol{A}_+ \in \mathbb{S}^{n}, \boldsymbol{b}_+ \in \mathbb{R}^{n}, c_+\in\mathbb{R}, \boldsymbol{\xi}_- \in \mathbb{R}^{m_-}.
     \end{split}
\end{equation}

and

$\text{(LSQTSVM2)}$
\begin{equation}\label{LSQTSVM2}
	\begin{split}
\min \ & \frac{1}{2}\sum_{j=1}^{m_{-}}(\frac{1}{2} {\boldsymbol{x}_{j}^{-}}^{\top} \boldsymbol{A}_- \boldsymbol{x}_{j}^{-}+\boldsymbol{b}_-^{\top}\boldsymbol{x}_{j}^{-}+c_-)^2 +C \sum_{i=1}^{m_{+}} (\xi_{i}^{+})^2, \\	
\text{s.t.} \ & \frac{1}{2} {\boldsymbol{x}_{i}^{+}}^{\top} \boldsymbol{A}_- \boldsymbol{x}_{i}^{+}+\boldsymbol{b}_-^{\top}\boldsymbol{x}_{i}^{+}+c_-=1-\xi_{i}^{+},\quad i=1,2, \ldots, m_{+},\\
& \boldsymbol{A}_- \in \mathbb{S}^{n}, \boldsymbol{b}_- \in \mathbb{R}^{n}, c_-\in\mathbb{R}, \boldsymbol{\xi}_+ \in \mathbb{R}^{m_+}.
     \end{split}
\end{equation}
where $C$ is a given positive parameter. $\boldsymbol{\xi}_{+}=\left[\xi_{1}^{+}, \ldots,
\xi_{m_{+}}^{+}\right]^{\top} \in \mathbb{R}^{m_{+}}$, $\boldsymbol\xi_{-}=\left[\xi_{1}^{-}, \ldots, \xi_{m_{-}}^{-}\right]^{\top} \in \mathbb{R}^{m_{-}}$. Take the objective function of the optimization problem (\ref{LSQTSVM1}) as an example and briefly introduce the constructive idea. Its first term is to make the positive points as close as possible to the positive quadratic surface $\frac{1}{2} \boldsymbol{x}^{\top} \boldsymbol{A}_+ \boldsymbol{x}+\boldsymbol{x}^{\top} \boldsymbol{b}_++c_+$, and its second term is to measure the loss of each negative point by the square loss function. The optimization problems (\ref{LSQTSVM1}) and (\ref{LSQTSVM2}) can be equivalently transformed by the definitions (\ref{dingyi2}) and (\ref{dingyi4}), then the corresponding solutions are obtained, as detailed in Ref\citep{gaoQSLSTSVM}.

If the solutions $\boldsymbol{A}_{\pm}$, $\boldsymbol{b}_{\pm}$ and $c_{\pm}$ of the optimization problems (\ref{LSQTSVM1}) and (\ref{LSQTSVM2}) are obtained, then  the label of a new sample $\boldsymbol{x}$ can be predicted by the following decision rule
\begin{equation}\label{juece}
{\rm Class}\ \boldsymbol{x}=\underset{k=+,-}{\operatorname{argmin}}\left\{\frac{\left|\frac{1}{2} \boldsymbol{x}^{\top} \boldsymbol{A}_k \boldsymbol{x}+\boldsymbol{x}^{\top} \boldsymbol{b}_k+c_k \right|}{\Vert\boldsymbol{A}_k \boldsymbol{x}+\boldsymbol{b}_k\Vert_{2}}\right\}.
    \end{equation}

\subsubsection{FRTBSVM}
As we all know the traditional TSVM model is lower than the classical SVM model in terms of computational consumption. However, the TSVM model uses the $L_2$ distance metric and hinge loss to increase the effect of outliers on the classifier. To improve this limitation of TSVM, Ma {\it et al.} proposed a robust and fast TBSVM model based on the capped $L_1$-norm distance metric called FRTBSVM \citep{FRTBSVM}. 

For a given binary dataset $\mathcal{T}$ (\ref{shujuD}), the kernel version of FRTBSVM tries to find the following pair of hyperplanes
\begin{equation}\label{chaoqumian}
\boldsymbol{u}_{\pm}^{\top} \boldsymbol{x}+b_{\pm}=0, 
\end{equation}
where $\boldsymbol{u}_{\pm}\in\mathbb{R}^{n}$, $b_{\pm}\in\mathbb{R}$. 

To find two hypersurfaces (\ref{chaoqumian}), the kernel version of the FRTBSVM model is presented as follows

$\text{(C${L}_1$FRTBSVM1)}$
\begin{equation}\label{KFRTBSVM1}
	\begin{split}
&\min \sum_{i=1}^{m_{+}} \min \left(\left|\boldsymbol{u}_{+}^{\top} \boldsymbol{x}_i^{+}+b_{+}\right|,  \varepsilon\right)+\frac{c}{2}\left(\left\|\boldsymbol{u}_{+}\right\|_{2}^{2}+b_{+}^{2}\right)+C\sum_{j=1}^{m_{-}} \min (|\eta_{j}^-|,  \varepsilon)\\
&\text { s.t. }-\left(\boldsymbol{X}_-\boldsymbol{u}_{+}+\boldsymbol{e}_-b_{+}\right)= \boldsymbol{e}_--\boldsymbol{\eta}_-, \\
&\qquad \ \boldsymbol{u}_{+}\in \mathbb{R}^{n}, b_{+}\in \mathbb{R}, \boldsymbol{\eta}_-\in\mathbb{R}^{m_-}, 
    \end{split}
\end{equation}
and

$\text{(C${L}_1$FRTBSVM2)}$
\begin{equation}\label{KFRTBSVM2}
	\begin{split}
&\min \sum_{j=1}^{m_{-}} \min \left(\left|\boldsymbol{u}_{-}^{\top} \boldsymbol{x}_j^{-}+b_{-}\right|,  \varepsilon\right)+\frac{c}{2}\left(\left\|\boldsymbol{u}_{-}\right\|_{2}^{2}+b_{-}^{2}\right)+C\sum_{i=1}^{m_{+}} \min (|\eta_{i}^+|,  \varepsilon)\\
&\text { s.t. }-\left(\boldsymbol{X}_+\boldsymbol{u}_{-}+\boldsymbol{e}_+b_{-}\right)= \boldsymbol{e}_+-\boldsymbol{\eta}_+, \\
&\qquad \ \boldsymbol{u}_{-}\in \mathbb{R}^{n}, b_{-}\in \mathbb{R}, \boldsymbol{\eta}_+\in\mathbb{R}^{m_+}, 
    \end{split}
\end{equation}
where $c$ and $C$ are known regularization parameters. The optimization problems (\ref{KFRTBSVM1}) and (\ref{KFRTBSVM2}) can be solved efficiently by an iterative algorithm. Specific details of the solution process and optimization problem formulation can be found in paper \citep{FRTBSVM}. If we obtain the solutions $\boldsymbol{u}_{\pm}$ and $b_{\pm}$, then the label of a new sample can be obtained by the following decision rule 
\begin{equation}\label{Kyuce}
{\rm Class}\ \boldsymbol{x}=\underset{k=+, -}{\operatorname{argmin}}\left\{\left|\boldsymbol{u}_{\pm}^{\top} \boldsymbol{x}+b_{\pm}\right|\right\} .
    \end{equation}
		\subsection{Capped $L_1$-norm}
As mentioned in the introduction, the capped $L_1$-norm loss is effective in reducing the effect of outliers on the classifier\citep{24,27,30,32}. Because outliers in the data usually have large errors, the capped $L_1$-norm usually helps the model filter out these points during the training process. However, using the capped $L_1$-norm distance metric usually results in the model becoming a non-smooth problem. Furthermore, the mathematical formulation of the capped $L_1$-norm loss is shown as follows
\begin{equation}
L_\epsilon(u)=min(|u|,\epsilon),
\end{equation}
where $\epsilon$ is a given parameter, $\epsilon=0.5$ in Figure \ref{Figure1}.


\begin{figure}[htbp]
\centering
\includegraphics[width=8cm]{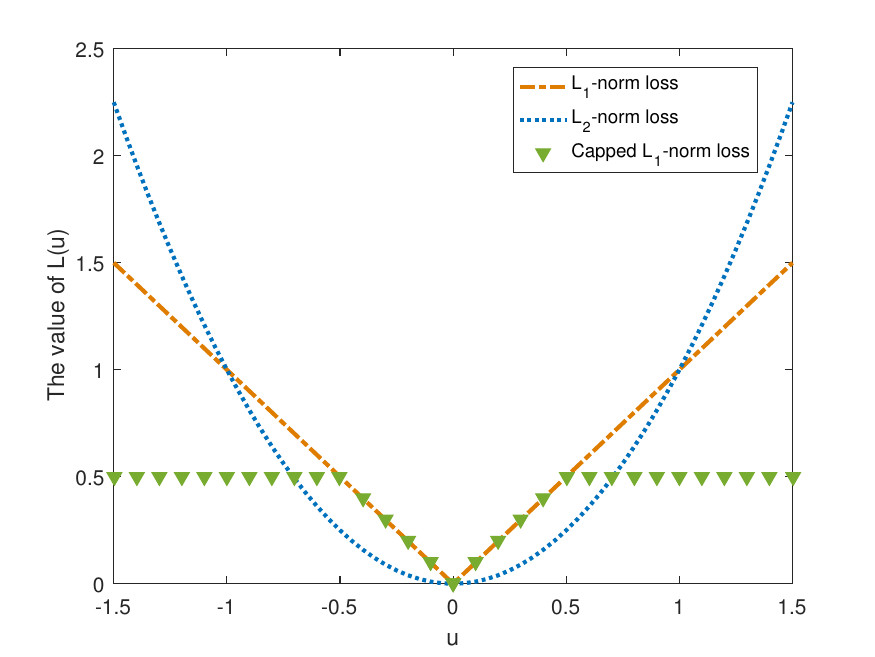}
\caption{The geometric interpretations of different norm metrics}\label{Figure1}
\end{figure}	

Figure \ref{Figure1} illustrates some geometric variations of the distance norm. As shown in Figure \ref{Figure1}, compared with the traditional $L_1$-norm loss and $L_2$-norm loss, the value of capped $L_1$-norm loss will remain constant when $u > \epsilon$. In other words, the capped $L_1$-norm loss can be used to measure noises and outliers. Therefore, capped $L_1$-norm loss is more robust than $L_1$-norm loss and $L_2$-norm loss.

\section{The proposed method}\label{3}
In this section, we first introduce the capped  $L_1$-norm quadric TSVM (C$L_1$QTSVM). Then, to solve the proposed models efficiently, we design fast iterative algorithms for it. Finally, we also analyze convergence, computational complexity and local optimal solutions of the algorithm in detail.
\subsection{C$L_1$QTSVM}
 For a given training set $\mathcal{T}$ (\ref{shujuD}), our goal is to also find two quadratic surfaces (\ref{erciqumian}), thus the mathematical expression for our C$L_1$QTSVM model is as follows

$\text{(C$L_1$QTSVM1)}$
\begin{equation}\label{CL1SLQTSVM1}
	\begin{split}
&\min \sum_{i=1}^{m_{+}} \min \left(\left|\frac{1}{2} {\boldsymbol{x}_{i}^{+}}^{\top} \boldsymbol{W}_+ \boldsymbol{x}_{i}^{+}+\boldsymbol{b}_+^{\top}\boldsymbol{x}_{i}^{+} +c_+ \right|, \varepsilon\right) +\frac{1}{2}c_1(\Vert\operatorname{hvec}(\boldsymbol{W}_{+})\Vert_{2}^{2}+\Vert\boldsymbol{b}_{+}\Vert_{2}^{2}+c_+^2)\\
&\qquad +c_2\sum_{j=1}^{m_{-}} \min \left(\left|\eta_j^-\right|, \varepsilon\right),\\
&\ \ \text{s.t.}  -\left(\frac{1}{2} {\boldsymbol{x}_{j}^{-}}^{\top} \boldsymbol{W}_+ \boldsymbol{x}_{j}^{-}+ \boldsymbol{b}_+^{\top}\boldsymbol{x}_{j}^{-}+c_+\right)=1-\eta_j^{-}, \quad j=1,2, \ldots, m_{-}, \\
&\; \qquad \boldsymbol{W}_+ \in \mathbb{S}^{n}, \boldsymbol{b}_+ \in \mathbb{R}^{n}, c_+\in\mathbb{R}, \boldsymbol{\eta}_- \in \mathbb{R}^{m_-},
     \end{split}
\end{equation}

and

$\text{(C$L_1$QTSVM2)}$
\begin{equation}\label{CL1SLQTSVM2}
	\begin{split}
&\min \sum_{j=1}^{m_{-}} \min \left(\left|\frac{1}{2} {\boldsymbol{x}_{j}^{-}}^{\top} \boldsymbol{W}_- \boldsymbol{x}_{j}^{-}+\boldsymbol{b}_-^{\top}\boldsymbol{x}_{j}^{-} +c_- \right|, \varepsilon\right) +\frac{1}{2}c_1(\Vert\operatorname{hvec}(\boldsymbol{W}_{-})\Vert_{2}^{2}+\Vert\boldsymbol{b}_{-}\Vert_{2}^{2}+c_-^2)\\
&\qquad +c_2 \sum_{i=1}^{m_{+}} \min \left(\left|\eta_i^{+}\right|, \varepsilon\right),\\
&\ \ \text{s.t.}\  \frac{1}{2} {\boldsymbol{x}_{i}^{+}}^{\top} \boldsymbol{W}_- \boldsymbol{x}_{i}^{+}+ \boldsymbol{b}_-^{\top}\boldsymbol{x}_{i}^{+}+c_-=1-\eta_i^{+}, \quad i=1,2, \ldots, m_{+}, \\
&\; \qquad \boldsymbol{W}_- \in \mathbb{S}^{n}, \boldsymbol{b}_- \in \mathbb{R}^{n}, c_-\in\mathbb{R}, \boldsymbol{\eta}_+ \in \mathbb{R}^{m_+},
     \end{split}
\end{equation}
 where $c_1$ and $c_2$ are known regularization parameters. $\varepsilon$ is a small values that is set to $10^{-5}$ in our experiments. 
 
 In the established C$L_1$-QTSVM model, the capped $L_1$-norm distance metric is used to weaken the effect of outliers on the kernel-free quadratic dual support vector machine model. Taking the objective function of the optimization problem (\ref{CL1SLQTSVM1}) as an example, we briefly introduce its construction principle. The first term of the objective function of the optimization problem (\ref{CL1SLQTSVM1}) is to make the positive points as close as possible to the positive quadratic decision hypersurface $\frac{1}{2} \boldsymbol{x}^{\top} \boldsymbol{W}_+ \boldsymbol{x}+\boldsymbol{x}^{ \top} \boldsymbol{b}_++c_+$. If the modulus of this positive point to a positive quadratic decision hypersurface is larger than $\varepsilon$, then set its metric to $\varepsilon$. The second term of the objective function of the optimization problem (\ref{CL1SLQTSVM1}) is the $L_2$ norm regularization term, which is used to improve the generalization ability of our model as well as to avoid the problem of non-existence of matrix inverses during the solution process. The third term of the objective function for the optimization problem (\ref{CL1SLQTSVM1}) is the loss term, whose value is determined by the location of the negative points. If the modulus of a negative point to a positive quadratic proximal hypersurface satisfies $\left|\frac{1}{2} \boldsymbol{x}_{j}^{\top} \boldsymbol{W}_+ \boldsymbol{x}_{j}^{-}+ \boldsymbol{b}_+^{\top}\boldsymbol{x}_{j}^{-}+c_++1\right|>\varepsilon$, the loss of the negative class point is set to $\varepsilon$. This helps to minimize the effect of outliers in the negative points on the positive quadratic decision hypersurfaces. The objective function of the optimization problem (\ref{CL1SLQTSVM2}) is built in a similar way and is not described here.

To further solve the optimization problems (\ref{CL1SLQTSVM1}) and (\ref{CL1SLQTSVM2}) , we define the following equations
\begin{equation}\label{Cgongshi1}
   \begin{split}
    \boldsymbol{w}_{\pm}\triangleq\left[\operatorname{hvec}^{\top}(\boldsymbol{W}_{\pm}), \boldsymbol{b}_{\pm}^{\top},c_{\pm}\right]^{\top}\in \mathbb{R}^{\frac{n^2+3n+2}{2}},
    \end{split}
\end{equation}

\begin{equation}\label{Cgongshi3}
   \begin{split}
   \boldsymbol{z}_{i}^{+}\triangleq\left[\operatorname{lvec}^{\top}(\boldsymbol{x}_{i}^{+}), {\boldsymbol{x}_{i}^{+}}^{\top}, 1\right]^{\top},\in \mathbb{R}^{\frac{n^2+3n+2}{2}},i=1, 2,\ldots, m_{+},
    \end{split}
\end{equation}

\begin{equation}\label{Cgongshi4}
   \begin{split}
   \boldsymbol{z}_{j}^{-}\triangleq\left[\operatorname{lvec}^{\top}(\boldsymbol{x}_{j}^{-}), {\boldsymbol{x}_{j}^{-}}^{\top}, 1\right]^{\top},\in \mathbb{R}^{\frac{n^2+3n+2}{2}},j=1, 2,\ldots, m_{-}.
    \end{split}
\end{equation}

By the definition of Eqs. (\ref{Cgongshi1})-(\ref{Cgongshi4}), we simplify the above two optimization problems as follows

$\text{(C$L_1$QTSVM1)}$
\begin{equation}\label{CL1SLQTSVM11}
	\begin{split}
\min \ & \ \sum_{i=1}^{m_{+}} \min \left(\left|\boldsymbol{w}_+^{\top}\boldsymbol{z}_i^+\right|, \varepsilon\right) +\frac{1}{2}c_1\Vert\boldsymbol{w}_{+}\Vert_{2}^{2}+c_2\sum_{j=1}^{m_{-}} \min \left(\left|\eta_j^-\right|, \varepsilon\right), \\	
\text{s.t.}& -\boldsymbol{w}_{+}^{\top}\boldsymbol{z}_{j}^{-}=1-\eta_{j}^{-},\quad j=1,2, \ldots, m_{-},\\
&\ \boldsymbol{w}_+ \in \mathbb{R}^{\frac{n^2+3n+2}{2}}, \boldsymbol{\eta}_- \in \mathbb{R}^{m_-},
     \end{split}
\end{equation}

and

$\text{(C$L_1$QTSVM2)}$
\begin{equation}\label{CL1SLQTSVM22}
	\begin{split}
\min \ & \ \sum_{j=1}^{m_{-}} \min \left(\left|\boldsymbol{w}_-^{\top}\boldsymbol{z}_j^+\right|, \varepsilon\right) +\frac{1}{2}c_1\Vert\boldsymbol{w}_{-}\Vert_{2}^{2}+c_2\sum_{i=1}^{m_{+}} \min \left(\left|\eta_i^+\right|, \varepsilon\right), \\	
\text{s.t.}& \; \boldsymbol{w}_{-}^{\top}\boldsymbol{z}_{i}^{+}=1-\eta_{i}^{+},\quad i=1,2, \ldots, m_{+},\\
&\ \boldsymbol{w}_- \in \mathbb{R}^{\frac{n^2+3n+2}{2}}, \boldsymbol{\eta}_+ \in \mathbb{R}^{m_+},
     \end{split}
\end{equation}

Furthermore, we note that it is difficult to solve the above optimization problems (\ref{CL1SLQTSVM11}) and (\ref{CL1SLQTSVM22}) by conventional convex optimization methods. By employing the re-weighted trick \citep{24,27,30}, we can reformulate the above two optimization problems into the following two approximate optimization problems

\text{(C$L_1$QTSVM1)}
\begin{equation}\label{minCL1SLQTSVM111}
  \begin{split}
 \min _{\boldsymbol{w}_{+},\boldsymbol{\eta}_{-}} \frac{1}{2}\left(\boldsymbol{G}_+^{\top}\boldsymbol{w}_+\right)^{\top}\boldsymbol{Q}\boldsymbol{G}_+^{\top}\boldsymbol{w}_+ +\frac{1}{2}c_1\Vert\boldsymbol{w}_{+}\Vert_{2}^{2}
+\frac{1}{2}c_2\boldsymbol{\eta}_-^{\top}\boldsymbol{U}\boldsymbol{\eta}_-,
  \end{split}
\end{equation}
where $\boldsymbol{Z}_{+}\triangleq\left[\boldsymbol{z}_{1}^{+},\ldots,\boldsymbol{z}_{m_+}^{+}\right]\in \mathbb{R}^{(\frac{n^2+3n+2}{2})\times m_+}$, and $\boldsymbol{Z}_{-}\triangleq\left[\boldsymbol{z}_{1}^{-},\ldots,\boldsymbol{z}_{m_-}^{-}\right]\in \mathbb{R}^{(\frac{n^2+3n+2}{2})\times m_-}$. $\boldsymbol{Q}\in\mathbb{R}^{m_+\times m_+}$ and $\boldsymbol{U}\in\mathbb{R}^{m_-\times m_-}$ are two diagonal matrices whose diagonal elements are defined as
\begin{equation}\label{q}
q_i=\left\{\begin{array}{ccc}
\frac{1}{\left|\boldsymbol{w}_{+}^{\top}\boldsymbol{z}_{i}^{+}\right|}, & \left|\boldsymbol{w}_{+}^{\top}\boldsymbol{z}_{i}^{+}\right|\leq \varepsilon,\\
\varepsilon, & {\rm otherwise}.
\end{array}\right. ,
\end{equation}
and
\begin{equation}\label{u}
u_j=\left\{\begin{array}{ccc}
\frac{1}{\left|\eta_j^{-}\right|}, & \left|\eta_j^{-}\right| \leq \varepsilon,\\
\varepsilon, & {\rm otherwise}.
\end{array}\right. .
\end{equation}

\text{(C$L_1$QTSVM2)}
\begin{equation}\label{minCL1SLQTSVM222}
  \begin{split}
 \min _{\boldsymbol{w}_{-},\boldsymbol{\eta}_{+}} \frac{1}{2}\left(\boldsymbol{G}_-^{\top}\boldsymbol{w}_-\right)^{\top}\boldsymbol{F}\boldsymbol{G}_-^{\top}\boldsymbol{w}_- +\frac{1}{2}c_1\Vert\boldsymbol{w}_{-}\Vert_{2}^{2}
+\frac{1}{2}c_2\boldsymbol{\eta}_+^{\top}\boldsymbol{G}\boldsymbol{\eta}_+,
  \end{split}
\end{equation}
where $\boldsymbol{F}\in\mathbb{R}^{m_-\times m_-}$ and $\boldsymbol{G}\in\mathbb{R}^{m_+\times m_+}$ are two diagonal matrices whose diagonal elements are defined as
\begin{equation}\label{f}
f_j=\left\{\begin{array}{ccc}
\frac{1}{\left|\boldsymbol{w}_{-}^{\top}\boldsymbol{z}_{j}^{-}\right|}, & \left|\boldsymbol{w}_{-}^{\top}\boldsymbol{z}_{j}^{-}\right|\leq \varepsilon,\\
\varepsilon, & {\rm otherwise}.
\end{array}\right. ,
\end{equation}
and
\begin{equation}\label{z}
g_i=\left\{\begin{array}{ccc}
\frac{1}{\left|\eta_i^{+}\right|}, & \left|\eta_i^{+}\right|  \leq \varepsilon,\\
\varepsilon, & {\rm otherwise}.
\end{array}\right. .
\end{equation}

\noindent\textbf{Remark 1.} It is worth noting that in the objective function of the optimization problem (\ref{minCL1SLQTSVM111}), the matrices $\boldsymbol{Q}$ and $\boldsymbol{U}$ are used to ``discard" outliers from the data. If positive points are far from the positive decision quadratic surface, then these points are considered positive outliers and are discarded. Similarly, if negative points are far from the positive proximal quadratic surface, they are also considered negative outliers and are excluded. The matrices $\boldsymbol{F}$ and $\boldsymbol{G}$ are interpreted similarly.

To solve the optimization problems (\ref{minCL1SLQTSVM111}) and (\ref{minCL1SLQTSVM222}), we first equivalently convert them into the following form
\begin{equation}\label{minCL1SLQTSVM1111}
  \begin{split}
\text{(C$L_1$QTSVM1)} \min _{\boldsymbol{w}_{+}} & \ J_+(\boldsymbol{w}_{+})\triangleq \frac{1}{2}\left(\boldsymbol{Z}_+^{\top}\boldsymbol{w}_+\right)^{\top}\boldsymbol{Q}\boldsymbol{Z}_+^{\top}\boldsymbol{w}_+ +\frac{1}{2}c_1\Vert\boldsymbol{w}_{+}\Vert_{2}^{2}\\
&+\frac{1}{2}c_2\left(\boldsymbol{Z}_-^{\top}\boldsymbol{w}_++\boldsymbol{e}_-\right)^{\top}\boldsymbol{U}\left(\boldsymbol{Z}_-^{\top}\boldsymbol{w}_++\boldsymbol{e}_-\right),
  \end{split}
\end{equation}
and
\begin{equation}\label{minCL1SLQTSVM2222}
  \begin{split}
\text{(C$L_1$QTSVM2)} \min _{\boldsymbol{w}_{-}} & \ J_-(\boldsymbol{w}_{-})\triangleq \frac{1}{2}\left(\boldsymbol{Z}_-^{\top}\boldsymbol{w}_-\right)^{\top}\boldsymbol{F}\boldsymbol{Z}_-^{\top}\boldsymbol{w}_- +\frac{1}{2}c_1\Vert\boldsymbol{w}_{-}\Vert_{2}^{2}\\
&+\frac{1}{2}c_2\left(\boldsymbol{e}_+-\boldsymbol{Z}_+^{\top}\boldsymbol{w}_-\right)^{\top}\boldsymbol{G}\left(\boldsymbol{e}_+-\boldsymbol{Z}_+^{\top}\boldsymbol{w}_-\right).
  \end{split}
\end{equation}

And to obtain the gradient for the optimization problems (\ref{minCL1SLQTSVM1111}) and (\ref{minCL1SLQTSVM2222}) and to set them to zero, we have
\begin{equation}\label{CL1SLQTSVM1111}
\nabla J_{+}\left(\boldsymbol{w}_{+}\right)= \boldsymbol{Z}_+\boldsymbol{Q}\boldsymbol{Z}_+^{\top}\boldsymbol{w}_{+}+c_1\boldsymbol{I}\boldsymbol{w}_{+}+c_2\boldsymbol{Z}_-\boldsymbol{U}\left(\boldsymbol{Z}_-^{\top}\boldsymbol{w}_++\boldsymbol{e}_-\right)=\boldsymbol{0},
\end{equation}
and
\begin{equation}\label{CL1SLQTSVM2222}
\nabla J_{-}\left(\boldsymbol{w}_{-}\right)= \boldsymbol{Z}_-\boldsymbol{F}\boldsymbol{Z}_-^{\top}\boldsymbol{w}_{-}+c_1\boldsymbol{I}\boldsymbol{w}_{-}-c_2\boldsymbol{Z}_+\boldsymbol{G}\left(\boldsymbol{e}_+-\boldsymbol{Z}_+^{\top}\boldsymbol{w}_-\right)=\boldsymbol{0}.
\end{equation}

To obtain the final iterative equations, we first define the following equations
\begin{equation}
\boldsymbol{M}_-^{t}\triangleq \boldsymbol{Y}^{t}-\boldsymbol{Y}^{t}\boldsymbol{Z}_-\left(\frac{1}{c_2}(\boldsymbol{U}^{t})^{-1}+\boldsymbol{Z}_-^{\top}\boldsymbol{Y}^{t}\boldsymbol{Z}_-\right)^{-1}\boldsymbol{Z}_-^{\top}\boldsymbol{Y}^{t},
\end{equation}
\begin{equation}
\boldsymbol{Y}^{t}\triangleq \frac{1}{c_1}\left(\boldsymbol{I}-\boldsymbol{Z}_+\left(c_1\boldsymbol{Q}^{t}+\boldsymbol{Z}_+^{\top}\boldsymbol{Z}_+\right)^{-1}\boldsymbol{Z}_+^{\top}\right),
\end{equation}
\begin{equation}
\boldsymbol{M}_+^{t}\triangleq \boldsymbol{H}^{t}-\boldsymbol{H}^{t}\boldsymbol{Z}_+\left(\frac{1}{c_2}(\boldsymbol{G}^{t})^{-1}+\boldsymbol{Z}_+^{\top}\boldsymbol{H}^{t}\boldsymbol{Z}_+\right)^{-1}\boldsymbol{Z}_+^{\top}\boldsymbol{H}^{t},
\end{equation}
\begin{equation}\label{45}
\boldsymbol{H}^{t}\triangleq \frac{1}{c_1}\left(\boldsymbol{I}-\boldsymbol{Z}_-\left(c_1\boldsymbol{F}^{t}+\boldsymbol{Z}_-^{\top}\boldsymbol{Z}_-\right)^{-1}\boldsymbol{Z}_-^{\top}\right).
\end{equation}

And then based on Eqs. (\ref{CL1SLQTSVM1111})-(\ref{45}) and Sherman-Morrison-Woodbury (SMW) theorem \citep{LSTSVM}, we propose the following iterative formulas with respect to the variables $\boldsymbol{w}_{+}$ and $\boldsymbol{w}_{-}$
\begin{equation}\label{jieCL1SLQTSVM11}
\boldsymbol{w}_{+}^{t+1}=\left\{\begin{array}{ccc}
-c_2\boldsymbol{M}_-^{t}\boldsymbol{Z}_-\boldsymbol{U}^{t}\boldsymbol{e}_-   , & \frac{n^2+3n+2}{2}>m_-,\\
 -c_2\left(\boldsymbol{Z}_+\boldsymbol{Q}^{t}\boldsymbol{Z}_+^{\top}+c_1\boldsymbol{I}+c_2\boldsymbol{Z}_-\boldsymbol{U}^{t}\boldsymbol{Z}_-^{\top}\right)^{-1}\boldsymbol{Z}_-\boldsymbol{U}^{t}\boldsymbol{e}_-   ,& {\rm otherwise},
\end{array}\right.,
\end{equation}
and
\begin{equation}\label{jieCL1SLQTSVM12}
\boldsymbol{w}_{-}^{t+1}=\left\{\begin{array}{ccc}
c_2\boldsymbol{M}_+^{t}\boldsymbol{Z}_+\boldsymbol{G}^{t}\boldsymbol{e}_+   , & \frac{n^2+3n+2}{2}>m_+,\\
 c_2\left(\boldsymbol{Z}_-\boldsymbol{F}^{t}\boldsymbol{Z}_-^{\top}+c_1\boldsymbol{I}+c_2\boldsymbol{Z}_+\boldsymbol{G}^{t}\boldsymbol{Z}_+^{\top}\right)^{-1}\boldsymbol{Z}_+\boldsymbol{G}^{t}\boldsymbol{e}_+   ,& {\rm otherwise},
\end{array}\right..
\end{equation}
where $t$ is the iteration number. By given initial values $\boldsymbol{w}_{+}^{0}$ and $\boldsymbol{w}_{-}^{0}$, the iterative process works by repeating Eqs. (\ref{jieCL1SLQTSVM11}) and (\ref{jieCL1SLQTSVM12}) until convergence. Thus, if solutions $\boldsymbol{w}_{+}$ and $\boldsymbol{w}_{-}$ are obtained, the two quadratic surfaces (\ref{erciqumian}) are also found. For a new sample $\boldsymbol{x}$, its label can be predicted by the decision function (\ref{juece}).

Then we briefly summarize the pseudo-code algorithmic framework for the C$L_1$QTSVM1 model as follows, and the algorithm for the C$L_1$QTSVM2 model is similar.
\begin{algorithm}
\caption{C$L_1$QTSVM1}\label{suanfaCL1SLDWPTSVM1}
\begin{algorithmic}[1]
\Require Training set $\mathcal{T}$ (\ref{shujuD}), convergence precision $\epsilon=10^{-8}$, maximum iteration number $N=30$, parameter $c_1$, $c_2$, threshold parameter $\varepsilon=10^{-8}$, model parameter $\boldsymbol{w}_{+}^{0}=\boldsymbol0_{(\frac{n^2+3n+2}{2})\times1}$, iteration number $t=0$.
\Ensure $\boldsymbol{W}_{+}$, $\boldsymbol{b}_{+}$,$c_+$.
\State   $\boldsymbol{z}_{i}^{+}\triangleq\left[\operatorname{hvec}^{\top}(\boldsymbol{x}_{i}), \boldsymbol{x}_{i}^{\top}, 1\right]^{\top}, \boldsymbol{z}_{j}^{-}\triangleq\left[\operatorname{hvec}^{\top}(\boldsymbol{x}_{j}), \boldsymbol{x}_{j}^{\top}, 1\right]^{\top}$.
\State $\boldsymbol{Z}_{+}\triangleq\left[\boldsymbol{z}_{1}^{+}, \ldots, \boldsymbol{z}_{m_+}^{+}\right]$,$\boldsymbol{Z}_{-}\triangleq\left[\boldsymbol{z}_{1}^{+}, \ldots, \boldsymbol{z}_{m_-}^{-}\right]$.
 \While{$t \leq N$}
  \For{$i \gets 1$ to $m_{+}$}
         \If{$\left|{\boldsymbol{z}_{i}^{+}}^{\top}\boldsymbol{w}_{+}^{t}\right|\leq \varepsilon$}
             \State $z_i^t\gets\frac{1}{\left|{\boldsymbol{z}_{i}^{+}}^{\top}\boldsymbol{w}_{+}^{t}\right|};$
        \Else
             \State  $z_i^t\gets\varepsilon.$
        \EndIf
 \EndFor
 \For{$j \gets 1$ to $m_{-}$}
      \If{$\left|1+{\boldsymbol{z}_{j}^{-}}^{\top}\boldsymbol{w}_{+}^{t}\right|\leq \varepsilon$}
             \State $u_j^t\gets\frac{1}{\left|1+{\boldsymbol{z}_{j}^{-}}^{\top}\boldsymbol{w}_{+}^{t}\right|};$
        \Else
             \State  $u_j^t\gets\varepsilon.$
        \EndIf
 \EndFor
 \State $\boldsymbol{Y}^{t}\gets \frac{1}{c_1}\left(\boldsymbol{I}-\boldsymbol{Z}_+\left(c_1\boldsymbol{Q}^{t}+\boldsymbol{Z}_+^{\top}\boldsymbol{Z}_+\right)^{-1}\boldsymbol{Z}_+^{\top}\right)$.
 \State $\boldsymbol{M}_-^{t}\gets \boldsymbol{Y}^{t}-\boldsymbol{Y}^{t}\boldsymbol{Z}_-\left(\frac{1}{c_2}(\boldsymbol{U}^{t})^{-1}+\boldsymbol{Z}_-^{\top}\boldsymbol{Y}^{t}\boldsymbol{Z}_-\right)^{-1}\boldsymbol{Z}_-^{\top}\boldsymbol{Y}^{t}.$
\If{$\frac{n^2+3n+2}{2}>m_-$}
             \State $\boldsymbol{w}_{+}^{t+1}\gets-c_2\boldsymbol{M}_-^{t}\boldsymbol{Z}_-\boldsymbol{U}^{t}\boldsymbol{e}_-;$
        \Else
             \State  $\boldsymbol{w}_{+}^{t+1}\gets-c_2\left(\boldsymbol{Z}_+\boldsymbol{Q}^{t}\boldsymbol{Z}_+^{\top}+c_1\boldsymbol{I}+c_2\boldsymbol{Z}_-\boldsymbol{U}^{t}\boldsymbol{Z}_-^{\top}\right)^{-1}\boldsymbol{Z}_-\boldsymbol{U}^{t}\boldsymbol{e}_- .$
        \EndIf
  \State$t\gets t+1$.
 \EndWhile
 \State $\left[\operatorname{hvec}^{\top}(\boldsymbol{W}_{+}), \boldsymbol{b}_{+}^{\top},c_+\right]^{\top}\gets\boldsymbol{w}_{+}^{t+1}.$
\end{algorithmic}
\end{algorithm}

\subsection{Convergence analysis}
In this subsection we analyze the convergence of the proposed Algorithm \ref{suanfaCL1SLDWPTSVM1} in detail. In particular, the convergence of Algorithm \ref{suanfaCL1SLDWPTSVM1} is given by Theorem \ref{dingli1}. To further prove Theorem \ref{dingli1} in detail, we first give the following lemma \ref{yinli1}.
\begin{lemma}\label{yinli1}
For any non-negative variables $x,y\in\mathbb{R}_+$, the following inequality holds:
\begin{equation}
\sqrt{x}-\frac{x}{2\sqrt{y}}\leq \sqrt{y}-\frac{y}{2\sqrt{y}}.
\end{equation}
\end{lemma}
\begin{proof}
 According to the inequality $\left(\sqrt{x}-\sqrt{y}\right)^2\geq0$,we have
 \begin{equation}
  \begin{split}
\left(\sqrt{x}-\sqrt{y}\right)^{2} \geq 0 & \Rightarrow x-2 \sqrt{x y}+y \geq 0 \Rightarrow \sqrt{x}-\frac{x}{2 \sqrt{y}} \leq \frac{\sqrt{y}}{2} \\
& \Rightarrow \sqrt{x}-\frac{x}{2\sqrt{y}} \leq \sqrt{y}-\frac{y}{2 \sqrt{y}}.
  \end{split}
 \end{equation}
\end{proof}

\begin{lemma}\label{yinli2}
For any non-negative variables $x,y,z,t,c\in\mathbb{R}_+$, the following inequality holds:
\begin{equation}
\sqrt{x}-\frac{x}{2\sqrt{y}}+c\left(\sqrt{z}-\frac{z}{2\sqrt{t}}\right)\leq \sqrt{y}-\frac{y}{2\sqrt{y}}+c\left(\sqrt{t}-\frac{t}{2\sqrt{t}}\right).
\end{equation}
\end{lemma}
\begin{proof}
From Lemma \ref{yinli1}, the following inequality holds
\begin{equation}\label{yinli22}
c\left(\sqrt{z}-\frac{z}{2\sqrt{t}}\right)\leq c\left(\sqrt{t}-\frac{t}{2\sqrt{t}}\right).
\end{equation}

Thus based on Lemma \ref{yinli1} and Eq. (\ref{yinli22}), we have
\begin{equation}
\sqrt{x}-\frac{x}{2\sqrt{y}}+c\left(\sqrt{z}-\frac{z}{2\sqrt{t}}\right)\leq \sqrt{y}-\frac{y}{2\sqrt{y}}+c\left(\sqrt{t}-\frac{t}{2\sqrt{t}}\right).
\end{equation}
\end{proof}

Next, we will prove the convergence of Algorithm \ref{suanfaCL1SLDWPTSVM1} based on the optimization problem (\ref{CL1SLQTSVM1}). Similarly, the proposed Algorithm \ref{suanfaCL1SLDWPTSVM1} to solve the optimization problem (\ref{CL1SLQTSVM2}) is also convergent.
\begin{theorem}\label{dingli1}
Algorithm \ref{suanfaCL1SLDWPTSVM1} will decrease the objective value of the optimization problem (\ref{CL1SLQTSVM1}) with each iteration until it converges.
\end{theorem}
\begin{proof}
First recall that the optimization problem (\ref{CL1SLQTSVM1}) is as follows
\begin{equation}\label{1CL1SLQTSVM1}
	\begin{split}
J&=\min_{\boldsymbol{W}_+,\boldsymbol{b}_+,c_+,\boldsymbol{\eta}_-} \sum_{i=1}^{m_{+}} \min \left(\left|\frac{1}{2} {\boldsymbol{x}_{i}^{+}}^{\top} \boldsymbol{W}_+ \boldsymbol{x}_{i}^{+}+\boldsymbol{b}_+^{\top}\boldsymbol{x}_{i}^{+} +c_+ \right|, \varepsilon\right)+c_2\sum_{j=1}^{m_{-}} \min \left(\left|\eta_j^-\right|, \varepsilon\right)\\
&\qquad \qquad \qquad +\frac{1}{2}c_1(\Vert\operatorname{hvec}(\boldsymbol{W}_{+})\Vert_{2}^{2}+\Vert\boldsymbol{b}_{+}\Vert_{2}^{2}+c_+^2),\\
&=\min_{\boldsymbol{w}_1^+,\boldsymbol{w}_2^+,\boldsymbol{\eta}_-} \sum_{i=1}^{m_{+}} \min \left(\left|{\boldsymbol{w}_1^+}^{\top}\boldsymbol{l}_i^++{\boldsymbol{w}_2^+}^{\top}\boldsymbol{s}_i^+\right|, \varepsilon\right) +c_2\sum_{j=1}^{m_{-}} \min \left(\left|\eta_j^-\right|, \varepsilon\right)\\ &\qquad \qquad \quad +\frac{1}{2}c_1\left(\Vert\boldsymbol{w}_{1}^{+}\Vert_{2}^{2}+\Vert\boldsymbol{w}_{2}^{+}\Vert_{2}^{2}\right),
     \end{split}
\end{equation}
where $\boldsymbol{w}_{1}^{+}\triangleq \operatorname{hvec}(\boldsymbol{W}_{+})\in \mathbb{R}^{\frac{n^2+n}{2}}$,$\boldsymbol{w}_{2}^{+}\triangleq\left[ \boldsymbol{b}_{+}^{\top},c_+\right]^{\top}\in \mathbb{R}^{n+1}$, $\boldsymbol{l}_{i}^{+}\triangleq \operatorname{lvec}(\boldsymbol{x}_{i}^{+})\in \mathbb{R}^{\frac{n^2+n}{2}}$,$\boldsymbol{s}_{i}^{+}\triangleq\left[ {\boldsymbol{x}_{i}^{+}}^{\top}, 1\right]^{\top}\in \mathbb{R}^{n+1},i=1, 2,\ldots, m_{+}$. Note that $\boldsymbol{w}_{1}^{+}$ contains nonlinear information in the data, while $\boldsymbol{w}_{2}^{+}$ contains linear information in the data; when $\boldsymbol{w}_{1}^{+}=\boldsymbol 0$, the optimization problem (\ref{CL1SLQTSVM1}) degenerates to a linear method, which reflects the nonlinear nature of our model.

When $\left|{\boldsymbol{w}_1^+}^{\top}\boldsymbol{l}_i^++{\boldsymbol{w}_2^+}^{\top}\boldsymbol{s}_i^+\right|<\varepsilon$ and $\left|\eta_j^-\right|<\varepsilon$, Eq. (\ref{1CL1SLQTSVM1}) is equivalent to
\begin{equation}\label{11CL1SLQTSVM1}
	\begin{split}
J=\min_{\boldsymbol{w}_1^+,\boldsymbol{w}_2^+,\boldsymbol{\eta}_-} \sum_{i=1}^{m_{+}} \left|{\boldsymbol{w}_1^+}^{\top}\boldsymbol{l}_i^++{\boldsymbol{w}_2^+}^{\top}\boldsymbol{s}_i^+\right|+c_2\sum_{j=1}^{m_{-}} \left|\eta_j^-\right| +\frac{1}{2}c_1\left(\Vert\boldsymbol{w}_{1}^{+}\Vert_{2}^{2}+\Vert\boldsymbol{w}_{2}^{+}\Vert_{2}^{2}\right).
     \end{split}
\end{equation}

We take the derivative of Eq. (\ref{11CL1SLQTSVM1}) with respect to $\boldsymbol{w}_1^+$, $\boldsymbol{w}_2^+$ and $\eta_j^-$ separately, we have
\begin{equation}\label{qiudao}
 \begin{split}
&\frac{\partial J}{\partial \boldsymbol{w}_1^+}=0\Rightarrow \sum_{i=1}^{m_{+}}\frac{\boldsymbol{l}_i^+\left({\boldsymbol{l}_i^+}^{\top}\boldsymbol{w}_1^++{\boldsymbol{s}_i^+}^{\top}\boldsymbol{w}_2^+\right)}{\left|{\boldsymbol{l}_i^+}^{\top}\boldsymbol{w}_1^++{\boldsymbol{s}_i^+}^{\top}\boldsymbol{w}_2^+\right|}+c_1\boldsymbol{w}_{1}^{+}=\boldsymbol{0},\\
&\frac{\partial J}{\partial \boldsymbol{w}_2^+}=0\Rightarrow \sum_{i=1}^{m_{+}}\frac{\boldsymbol{s}_i^+\left({\boldsymbol{l}_i^+}^{\top}\boldsymbol{w}_1^++{\boldsymbol{s}_i^+}^{\top}\boldsymbol{w}_2^+\right)}{\left|{\boldsymbol{l}_i^+}^{\top}\boldsymbol{w}_1^++{\boldsymbol{s}_i^+}^{\top}\boldsymbol{w}_2^+\right|}+c_1\boldsymbol{w}_{2}^{+}=\boldsymbol{0},\\
&\frac{\partial J}{\partial\eta_j^-} =0\Rightarrow c_2\sum_{j=1}^{m_{-}}\frac{\eta_j^-}{\left|\eta_j^-\right|}=0.
 \end{split}
\end{equation}

Then combining the above Eqs. (\ref{qiudao}), we have
\begin{equation}\label{qiudao2}
  \begin{split}
&\sum_{i=1}^{m_{+}}\frac{\boldsymbol{l}_i^+\left({\boldsymbol{l}_i^+}^{\top}\boldsymbol{w}_1^++{\boldsymbol{s}_i^+}^{\top}\boldsymbol{w}_2^+\right)}{\left|{\boldsymbol{l}_i^+}^{\top}\boldsymbol{w}_1^++{\boldsymbol{s}_i^+}^{\top}\boldsymbol{w}_2^+\right|}+\sum_{i=1}^{m_{+}}\frac{\boldsymbol{s}_i^+\left({\boldsymbol{l}_i^+}^{\top}\boldsymbol{w}_1^++{\boldsymbol{s}_i^+}^{\top}\boldsymbol{w}_2^+\right)}{\left|{\boldsymbol{l}_i^+}^{\top}\boldsymbol{w}_1^++{\boldsymbol{s}_i^+}^{\top}\boldsymbol{w}_2^+\right|}+c_1\left(\boldsymbol{w}_{1}^{+}+\boldsymbol{w}_{2}^{+}\right)+c_2\sum_{j=1}^{m_{-}}\frac{\eta_j^-}{\left|\eta_j^-\right|}\\
&=\sum_{i=1}^{m_{+}}\frac{\boldsymbol{z}_i^+{\boldsymbol{z}_i^+}^{\top}\boldsymbol{w}_+}{\left|\boldsymbol{w}_+^{\top}\boldsymbol{z}_i^+\right|}+c_1\boldsymbol{w}_++c_2\sum_{j=1}^{m_{-}}\frac{\eta_j^-}{\left|\eta_j^-\right|}=0,
  \end{split}
\end{equation}
where $\boldsymbol{w}_{+}\triangleq\left[{\boldsymbol{w}_{1}^{+}}^{\top},{\boldsymbol{w}_{2}^{+}}^{\top}\right]^{\top}\in \mathbb{R}^{\frac{n^2+3n+2}{2}}$, $\boldsymbol{z}_{i}^{+}\triangleq\left[{\boldsymbol{l}_{i}^{+}}^{\top},{\boldsymbol{s}_{i}^{+}}^{\top}\right]^{\top}\in \mathbb{R}^{\frac{n^2+3n+2}{2}}$.

Then we construct two diagonal matrices $\boldsymbol{Q}$ and $\boldsymbol{U}$, and define $q_i=\frac{1}{\left|\boldsymbol{w}_+^{\top}\boldsymbol{z}_i^+\right|}$ and $u_j=\frac{1}{\left|\eta_j^-\right|}$ as the diagonal elements of $\boldsymbol{Q}$ and $\boldsymbol{U}$ respectively. Rewriting Eq. (\ref{qiudao2}) leads to
\begin{equation}\label{qiudao3}
\boldsymbol{Z}_+\boldsymbol{Q}\boldsymbol{Z}_+^{\top}\boldsymbol{w}_{+}+c_1\boldsymbol{I}\boldsymbol{w}_{+}+c_2\boldsymbol{U}\boldsymbol{\eta}_-=0.
\end{equation}

Apparently, Eqs. (\ref{qiudao3}) is the optimal solution to the following optimization problem
\begin{equation}
 \min _{\boldsymbol{w}_{+},\boldsymbol{\eta}_{-}} \frac{1}{2}\left(\boldsymbol{Z}_+^{\top}\boldsymbol{w}_+\right)^{\top}\boldsymbol{Q}\boldsymbol{Z}_+^{\top}\boldsymbol{w}_+ +\frac{1}{2}c_1\Vert\boldsymbol{w}_{+}\Vert_{2}^{2}
+\frac{1}{2}c_2\boldsymbol{\eta}_-^{\top}\boldsymbol{U}\boldsymbol{\eta}_-.
\end{equation}

Now, let $\overline{\boldsymbol{w}}_{+}\triangleq\left[\operatorname{hvec}^{\top}(\overline{\boldsymbol{W}}_{+}), \overline{\boldsymbol{b}}_{+}^{\top},\overline{c}_+\right]^{\top}$ and $\overline{\eta}_-\triangleq\boldsymbol{Z}_-^{\top}\overline{\boldsymbol{w}}_++\boldsymbol{e}_-$ represent the updated $\boldsymbol{w}_{+}$ and $\boldsymbol{\eta}_{-}$ by Eqs. (\ref{jieCL1SLQTSVM11}). Then it is easy to obtain the following inequality
\begin{equation}\label{budengshi}
 \begin{split}
&\frac{1}{2}\overline{\boldsymbol{w}}_+^{\top}\boldsymbol{Z}_+\boldsymbol{Q}\boldsymbol{Z}_+^{\top}\overline{\boldsymbol{w}}_+ +\frac{1}{2}c_1\Vert\overline{\boldsymbol{w}}_{+}\Vert_{2}^{2}
+\frac{1}{2}c_2\overline{\boldsymbol{\eta}}_-^{\top}\boldsymbol{U}\overline{\boldsymbol{\eta}}_-\\
&\leq\frac{1}{2}\boldsymbol{w}_+^{\top}\boldsymbol{Z}_+\boldsymbol{Q}\boldsymbol{Z}_+^{\top}\boldsymbol{w}_+ +\frac{1}{2}c_1\Vert\boldsymbol{w}_{+}\Vert_{2}^{2}
+\frac{1}{2}c_2\boldsymbol{\eta}_-^{\top}\boldsymbol{U}\boldsymbol{\eta}_-.
 \end{split}
\end{equation}

Next, rewriting Eqs. (\ref{budengshi}) in component form, we can obtain the following inequality
\begin{equation}\label{budengshi1}
 \begin{split}
&\sum_{i=1}^{m_{+}}\frac{({\boldsymbol{z}_i^+}^{\top}\overline{\boldsymbol{w}}_+)^{\top}{\boldsymbol{z}_i^+}^{\top}\overline{\boldsymbol{w}}_+}{2|\boldsymbol{w}_+^{\top}\boldsymbol{z}_i^+|}+\frac{1}{2}c_1\Vert\overline{\boldsymbol{w}}_{+}\Vert_{2}^{2}+\frac{1}{2}c_2\sum_{j=1}^{m_{-}}\frac{(\overline{\eta}_j^-)^2}{\left|\eta_j^-\right|}\\
&\leq \sum_{i=1}^{m_{+}}\frac{({\boldsymbol{z}_i^+}^{\top}\boldsymbol{w}_+)^{\top}{\boldsymbol{z}_i^+}^{\top}\boldsymbol{w}_+}{2|\boldsymbol{w}_+^{\top}\boldsymbol{z}_i^+|}+\frac{1}{2}c_1\Vert\boldsymbol{w}_{+}\Vert_{2}^{2}+\frac{1}{2}c_2\sum_{j=1}^{m_{-}}\frac{(\eta_j^-)^2}{|\eta_j^-|}.
 \end{split}
\end{equation}

Furthermore, let $\sqrt{x}=|{\boldsymbol{z}_i^+}^{\top}\overline{\boldsymbol{w}}_+|$, $\sqrt{y}=|{\boldsymbol{z}_i^+}^{\top}\boldsymbol{w}_+|$, $\sqrt{z}=|\overline{\eta}_j^-|$,$\sqrt{t}=|\eta_j^-|$,$c=c_2$. From Lemma \ref{yinli2}, we have the following inequality
\begin{equation}\label{1budengshi1}
 \begin{split}
 &|{\boldsymbol{z}_i^+}^{\top}\overline{\boldsymbol{w}}_+|-\frac{({\boldsymbol{z}_i^+}^{\top}\overline{\boldsymbol{w}}_+)^{\top}({\boldsymbol{z}_i^+}^{\top}\overline{\boldsymbol{w}}_+)}{2|{\boldsymbol{z}_i^+}^{\top}\boldsymbol{w}_+|}+c_2\left(|\overline{\eta}_j^-|-\frac{(\overline{\eta}_j^-)^2}{|\eta_j^-|}\right)\\
 &\leq|{\boldsymbol{z}_i^+}^{\top}\boldsymbol{w}_+|-\frac{({\boldsymbol{z}_i^+}^{\top}\boldsymbol{w}_+)^{\top}({\boldsymbol{z}_i^+}^{\top}\boldsymbol{w}_+)}{2|{\boldsymbol{z}_i^+}^{\top}\boldsymbol{w}_+|}+c_2\left(|\eta_j^-|-\frac{(\eta_j^-)^2}{|\eta_j^-|}\right).
 \end{split}
\end{equation}

Thus based on Eq. (\ref{1budengshi1}), we can further obtain the following inequality
\begin{equation}\label{11budengshi1}
 \begin{split}
 &\sum_{i=1}^{m_{+}}|{\boldsymbol{z}_i^+}^{\top}\overline{\boldsymbol{w}}_+|-\sum_{i=1}^{m_{+}}\frac{({\boldsymbol{z}_i^+}^{\top}\overline{\boldsymbol{w}}_+)^{\top}({\boldsymbol{z}_i^+}^{\top}\overline{\boldsymbol{w}}_+)}{2|{\boldsymbol{z}_i^+}^{\top}\boldsymbol{w}_+|}+c_2\sum_{j=1}^{m_{-}}|\overline{\eta}_j^-|-c_2\sum_{j=1}^{m_{-}}\frac{(\overline{\eta}_j^-)^2}{|\eta_j^-|}\\
 &\leq\sum_{i=1}^{m_{+}}|{\boldsymbol{z}_i^+}^{\top}\boldsymbol{w}_+|-\sum_{i=1}^{m_{+}}\frac{({\boldsymbol{z}_i^+}^{\top}\boldsymbol{w}_+)^{\top}({\boldsymbol{z}_i^+}^{\top}\boldsymbol{w}_+)}{2|{\boldsymbol{z}_i^+}^{\top}\boldsymbol{w}_+|}+c_2\sum_{j=1}^{m_{-}}|\eta_j^-|-c_2\sum_{j=1}^{m_{-}}\frac{(\eta_j^-)^2}{|\eta_j^-|}.
 \end{split}
\end{equation}

Then further combining Eqs. (\ref{budengshi1}) and (\ref{11budengshi1}), we have
\begin{equation}\label{111budengshi1}
 \begin{split}
 &\sum_{i=1}^{m_{+}}|{\boldsymbol{z}_i^+}^{\top}\overline{\boldsymbol{w}}_+|+\frac{1}{2}c_1\Vert\overline{\boldsymbol{w}}_{+}\Vert_{2}^{2}+c_2\sum_{j=1}^{m_{-}}|\overline{\eta}_j^-|\\
 &\leq\sum_{i=1}^{m_{+}}|{\boldsymbol{z}_i^+}^{\top}\boldsymbol{w}_+|+\frac{1}{2}c_1\Vert\boldsymbol{w}_{+}\Vert_{2}^{2}+c_2\sum_{j=1}^{m_{-}}|\eta_j^-|.
 \end{split}
\end{equation}

Thus we can obtain the following inequality
\begin{equation}\label{1111budengshi1}
 \begin{split}
 &\sum_{i=1}^{m_{+}}\min\left(|{\boldsymbol{z}_i^+}^{\top}\overline{\boldsymbol{w}}_+|, \varepsilon\right)+\frac{1}{2}c_1\Vert\overline{\boldsymbol{w}}_{+}\Vert_{2}^{2}+c_2\sum_{j=1}^{m_{-}}\min\left(|\overline{\eta}_j^-|, \varepsilon\right)\\
 &\leq\sum_{i=1}^{m_{+}}\min\left(|{\boldsymbol{z}_i^+}^{\top}\boldsymbol{w}_+|, \varepsilon\right)+\frac{1}{2}c_1\Vert\boldsymbol{w}_{+}\Vert_{2}^{2}+c_2\sum_{j=1}^{m_{-}}\min\left(|\eta_j^-|, \varepsilon\right).
 \end{split}
\end{equation}
\end{proof}
Thus we have $J(\overline{\boldsymbol{w}}_+,\overline{\boldsymbol{\eta}}_-)\leq J(\boldsymbol{w}_+,\boldsymbol{\eta}_-)$ established. Similarly, when ${\boldsymbol{z}_i^+}^{\top}\boldsymbol{w}_+>\varepsilon$ and $\eta_j^->\varepsilon$ are satisfied, we have $J(\overline{\boldsymbol{w}}_+,\overline{\boldsymbol{\eta}}_-)=J(\boldsymbol{w}_+,\boldsymbol{\eta}_-)$. In summary, $J(\overline{\boldsymbol{w}}_+,\overline{\boldsymbol{\eta}}_-)\leq J(\boldsymbol{w}_+,\boldsymbol{\eta}_-)$ holds. This means that Algorithm \ref{suanfaCL1SLDWPTSVM1} will decrease the objective function of the optimization problem (\ref{CL1SLQTSVM1}) until it converges. The proof process is similar for the optimization problem (\ref{CL1SLQTSVM2}).

\begin{theorem}\label{dingli2}
Algorithm \ref{suanfaCL1SLDWPTSVM1} will converge to a local optimum to the optimization problem (\ref{CL1SLQTSVM11}).
\end{theorem}
\begin{proof}
Here we use optimization problem (\ref{CL1SLQTSVM11}) as an example to prove the Theorem \ref{dingli2}. When $|\boldsymbol{w}_+^{\top}\boldsymbol{z}_i^+|<\varepsilon$ and $|\eta_j^-|<\varepsilon$, the Lagrangian function of the optimization problem (\ref{CL1SLQTSVM11}) is given as follows
\begin{equation}\label{CCL1SLQTSVM11}
	\begin{split}
&L_1(\boldsymbol{w}_+,\boldsymbol{\eta}_-,\boldsymbol{\alpha})=\sum_{i=1}^{m_{+}} |\boldsymbol{w}_+^{\top}\boldsymbol{z}_i^+| +\frac{1}{2}c_1\Vert\boldsymbol{w}_{+}\Vert_{2}^{2}+c_2\sum_{j=1}^{m_{-}} |\eta_j^-|\\
 &\qquad \qquad \quad \ \ -\boldsymbol{\alpha}^{\top}(-\boldsymbol{Z}_{-}^{\top}\boldsymbol{w}_{+} +\boldsymbol{\eta}_{-}-\boldsymbol{e}_{-}).
     \end{split}
\end{equation}

Then to obtain the gradient of $L_1(\boldsymbol{w}_+,\boldsymbol{\eta}_-,\boldsymbol{\alpha})$ with respect to $\boldsymbol{w}_+$, we have
\begin{equation}\label{qiudao1111}
 \begin{split}
\frac{\partial L_1(\boldsymbol{w}_+,\boldsymbol{\eta}_-,\boldsymbol{\alpha})}{\partial \mathbf{\boldsymbol{w}}_{+}}&= \sum_{i=1}^{m_{+}}\frac{\boldsymbol{z}_i^+{\boldsymbol{z}_i^+}^{\top}\boldsymbol{w}_+}{\left|\boldsymbol{w}_+^{\top}\boldsymbol{z}_i^+\right|}+c_1\boldsymbol{I}\boldsymbol{w}_++\boldsymbol{Z}_-\boldsymbol{\alpha}\\
&=\boldsymbol{Z}_+\boldsymbol{Q}\boldsymbol{Z}_+^{\top}\boldsymbol{w}_{+}+c_1\boldsymbol{I}\boldsymbol{w}_++\boldsymbol{Z}_-\boldsymbol{\alpha}=\boldsymbol{0}.
 \end{split}
\end{equation}

Similarly, we can obtain the Lagrangian function of the optimization problem (\ref{minCL1SLQTSVM111}) as follows
\begin{equation}\label{1222minCL1SLQTSVM111}
  \begin{split}
 L_2(\boldsymbol{w}_{+},\boldsymbol{\eta}_-,\boldsymbol{\alpha})= &\frac{1}{2}\left(\boldsymbol{Z}_+^{\top}\boldsymbol{w}_+\right)^{\top}\boldsymbol{Q}\boldsymbol{Z}_+^{\top}\boldsymbol{w}_+ +\frac{1}{2}c_1\Vert\boldsymbol{w}_{+}\Vert_{2}^{2}
+\frac{1}{2}c_2\boldsymbol{\eta}_-^{\top}\boldsymbol{U}\boldsymbol{\eta}_-\\
&-\boldsymbol{\alpha}^{\top}(-\boldsymbol{Z}_-^{\top}\boldsymbol{w}_++\boldsymbol{\eta}_--\boldsymbol{e}_-).
  \end{split}
\end{equation}

Then the gradient of $L_2(\boldsymbol{w}_{+},\boldsymbol{\eta}_-,\boldsymbol{\alpha})$ with respect to $\boldsymbol{w}_{+}$ yields
\begin{equation}\label{qiudao111111}
 \begin{split}
\frac{\partial L_2(\boldsymbol{w}_+,\boldsymbol{\eta}_-,\boldsymbol{\alpha})}{\partial \mathbf{\boldsymbol{w}}_{+}}\Rightarrow \boldsymbol{Z}_+\boldsymbol{Q}\boldsymbol{Z}_+^{\top}\boldsymbol{w}_{+}+c_1\boldsymbol{I}\boldsymbol{w}_++\boldsymbol{Z}_-\boldsymbol{\alpha}=\boldsymbol{0}.
 \end{split}
\end{equation}
\end{proof}

It is easy to observe that Eqs. (\ref{qiudao1111}) and (\ref{qiudao111111}) are equivalent when Algorithm \ref{suanfaCL1SLDWPTSVM1} converges, which satisfies the Lagrangian function condition of the optimization problem (\ref{CL1SLQTSVM11}). Thus we can solve the optimization problem (\ref{minCL1SLQTSVM111}) instead of the optimization problem (\ref{CL1SLQTSVM11}). This means that the convergent solution satisfies Eq. (\ref{qiudao111111}) when it is a locally optimal solution of the optimization problem (\ref{CL1SLQTSVM11}). A similar proof is available for problem (\ref{CL1SLQTSVM22}).

\subsection{Time complexity}
In this subsection, we briefly analyze the time complexity of Algorithm \ref{suanfaCL1SLDWPTSVM1} for the optimization problem (\ref{CL1SLQTSVM1}). Let $m_+$ and $m_-$ be the number of positive and negative samples, respectively. $n$ represents the feature number of the samples and let $m_l=\frac{n^2+3n+2}{2}$. And $T$ represents the maximum number of iterations, it should be noted that $T=30$ in our experiments. The time complexity of Algorithm \ref{suanfaCL1SLDWPTSVM1} is mainly in the matrix inverse operation, thus we need to analyze its time complexity in the following two cases. When $m_l\geq m_-$, the time complexity of Algorithm \ref{suanfaCL1SLDWPTSVM1} is $O(Tm_-^3)$. When $m_l<m_-$, the time complexity of Algorithm \ref{suanfaCL1SLDWPTSVM1} is $O(Tm_l^3)$. Thus, the time complexity of Algorithm \ref{suanfaCL1SLDWPTSVM1} for the optimization problem (\ref{CL1SLQTSVM1}) is $O(Tmin(m_l^3,m_-^3))$. In summary, our C$L_1$QTSVM model requires time complexity of $O(T(min(m_l^3,m_-^3)+min(m_l^3,m_+^3)))$.

\subsection{Comparison methods}

    To better distinguish it from the comparative state-of-the-art methods, our model is compared with the above methods, as shown in Table \ref{table22}. As a result, it can be summarized that our method has the following four major advantages. (1) Our model uses the Capped $L_1$-norm distance metric, which results in its insensitivity to outliers. (2) Our model does not involve the selection of kernel function and kernel parameters, and has stronger interpretability. (3) The introduction of the $L_2$ regularization term leads our method to follow the structural risk minimization (SRM) principle, and further improves the generalization ability of our model. (4) An efficient iterative algorithm is used to solve the optimization model, resulting in low time cost.
    
  \renewcommand\tabcolsep{15pt}
\begin{table*}[htbp]
\centering
\caption{ Comparison between different methods.}\label{table22}
\vspace{-0.3cm}
\resizebox{1\columnwidth}{1.2cm}{
\begin{tabular}{llllll}
\hline
Model & Loss function & Distance norm& $L_2$ regularization term & Kernel function & Time complexity\\
\hline
   TSVM \citep{TSVM}& Hinge loss& $L_2$-norm&No&Yes&$O(\frac{m^3}{4}+(m+1)^{3})$\\
    C$L_1$-FRTBSVM \citep{FRTBSVM}& Capped $L_1$-norm& Capped $L_1$-norm&Yes&Yes&$O(2T(m+1)^{3})$\\
    C$L_1$-FRTELM \citep{L1TELM}& Capped $L_1$-norm& Capped $L_1$-norm&No&No&$O(2(L^{3}+n^{3}))$\\
    LINEX-TSVM \citep{LINEX-TSVM}& Capped LINEX loss& Capped $L_1$-norm&Yes&Yes&$O(T(\frac{m^3}{4}+(n+1)^{3}))$\\
    C$L_{2,p}$-LSTSVM \citep{CL2p}& Square loss& Capped $L_{2,p}$-norm&No&Yes&$O(T(m\times (n+1)^{2}+(n+1)^{3}))$\\
    $\nu$-FRSQSSVM \citep{gaonu-FRSQSSVM}&Fuzzy hinge loss&Fuzzy $L_2$-norm&No&No&$O(m^{3}+(2n+1)^{3})$\\
     LSQTSVM \citep{gaoQSLSTSVM}& Square loss&$L_2$-norm&No&No&$O(\frac{n^2+3n+2}{2}\times \frac{n^2+3n+2}{2})$\\
     C$L_1$QTSVM & Capped $L_1$-norm&Capped $L_1$-norm&Yes&No&$O(T(min(m_l^3,m_-^3)+min(m_l^3,m_+^3)))$\\
\hline
\end{tabular}
}
\end{table*}

\section{Numerical experiments}\label{4}
In this section, we first perform experiments on synthetic dataset, benchmark dataset and image dataset, and compare our method with the rest seven state-of-the-art methods. Then, parametric sensitivity analysis, convergence analysis and nonparametric tests are also performed to verify the effectiveness of our method.

\subsection{Experiment setups}
(1) Running environment. All experiments were run on a 16GB RAM and Inter Core i7-4790 desktop computer.

(2) Benchmark dataset description.
\renewcommand\tabcolsep{30pt}
\begin{table*}[htbp]
\centering
\caption{Description of the benchmark datasets.}\label{jizhun1}
\vspace{-0.3cm}
\resizebox{.7\columnwidth}{!}{
\begin{tabular}{lll}
\hline
Dataset &\#Instances&\#Attributes \\
\hline
   Australian&690&14\\
		                       Indian&583&10\\
		                      Blood&748&4\\
		                       Cleve&296&13\\
		                      Climate-simulation&540&18\\
		CMC&1473&9\\
		Cylinder-bands&512&35\\
		Glass(0-1-5 vs. 2)&192&9\\
		Haberman&306&3\\
		Hepatitis&155&19\\
		Ionosphere&351&32\\
	    Sonar&208&60\\
		Vote&435&16\\
		WBC&683&9\\
		Yeast(2 vs. 8)&264&8\\
		Robotnavig&2923&24\\
\hline
\end{tabular}
}
\end{table*}

\begin{enumerate}
  \item[$\bullet$] In the synthetic numerical experiments, we use three synthetic datasets, Example 1, Example 2 and Example 3, as shown in Figure \ref{example2}. To further compare the classification performance of each method under different label noise ratios, we add 5$\%$ and 10$\%$ label noise to Example 1, Example 2 and Example 3 respectively.
  \item[$\bullet$] In the numerical experiments on real data, we first use 16 benchmark datasets, which are all from the UCI database, and the details of these 16 benchmark datasets are listed in Table \ref{jizhun1}. In addition, image datasets are also used, and their details are described in the following subsections.
    \end{enumerate}

(3) Data preprocessing. Before the experiment, all data are normalized to $[-1,1]$.

(4) Parameter setting. In all numerical experiments, the 5-fold cross-validation method and grid search method were used to select the optimal parameters for each model. The final classification result is the average of 10 times 10-fold cross-validation. For the nonlinear case, the Gaussian kernel $\boldsymbol{\mathcal{K}}(\boldsymbol{u}, \boldsymbol{v})=\exp \left(-\frac{\|\boldsymbol{u}-\boldsymbol{v}\|^{2}}{2\sigma^{2}}\right)$ is used, the kernel parameter $\sigma$ is chosen from the set $\left\{2^{i}|i=-5,-4, . ... ,4,5\right\}$. The loss function regularization parameter $c_2$ and the regularization parameter $c_1$ of all methods are selected from the set $\left\{10^{i}|i=-5,-4,...,4,5\right\}$. In addition, the specific parameter selection details of some comparison methods are as follows

\begin{enumerate}
  \item[$\bullet$] For FRTELM model \citep{L1TELM}, hidden layer node parameter $L \in \left\{20,40,50,100,200\right\}$.
   \item[$\bullet$] For the LINEX-TSVM model \citep{LINEX-TSVM}, Capped LINEX loss function parameter $a=1$, convergence precision is set to $\epsilon=10^{-4}$.
    \item[$\bullet$] For C$L_{2,p}$LSTSVM model \citep{CL2p}, the truncation parameter $\epsilon$ is selected from the set $\left\{10^{i}|i=-5,-3, ... ,3,5\right\}$, with parameters $p$ selected from the set $\left\{1.1,1.3,... ,1.9\right\}$, the convergence accuracy is set to $\epsilon=10^{-4}$.
\item[$\bullet$] For $\nu$-FRSQSSVM model \citep{gaonu-FRSQSSVM}, the parameters $\nu$ is chosen from the set $\left\{0.05,0.1,... ,0.9,0.99\right\}$.
\end{enumerate}

(5) Evaluation metrics . The experimental part of this section uses Accuracy (Acc) and F1 Score (F1) to evaluate the classification performance of each method, which are defined as follows
\begin{equation}\label{AUC}
{\rm Acc}=\frac{TP+TN}{TP+FN+TN+FP},
\end{equation}
\begin{equation}\label{AUC}
{\rm F1}=\frac{2TP}{2TP+FN+FP},
\end{equation}
where $TN$ and $TP$ represent the number of correctly predicted positive and negative samples, respectively. And $FP$ and $FN$ represent the number of incorrectly predicted positive and negative samples, respectively. In addition, the standard deviation of the test accuracy ($Std$) is also used.

\subsection{Experiments results}
In this subsection, we focus on comparing the classification performance of our method with the rest of the state-of-the-art methods on synthetic datasets, benchmark datasets, and image datasets.

\subsubsection{Performance on synthetic dataset}

First, we compare the classification performance of the C$L_1$QTSVM model with other state-of-the-art methods on three synthetic datasets. The construction methods and images of the three synthetic datasets are shown as follows:

\textbf{Example 1}:\\
class +1: $[x_i]_2=0.2222[x_i]_1^{2}+0.5+\xi_i$,$[x_i]_1 \sim U[-3,3]$,\\
class -1: $[x_j]_2=-0.2222[x_j]_1^{2}+1.5+\xi_j$, $[x_j]_1 \sim U[-3,3]$,\\
where  $\xi_i,\xi_j \sim N(0,0.1)$,$i,j=1,2.... ,200$.

\begin{figure}[htbp]
\centering
\subfigure[Example 1]{
\begin{minipage}[t]{0.3\linewidth}
\centering
\includegraphics[width=4.5cm]{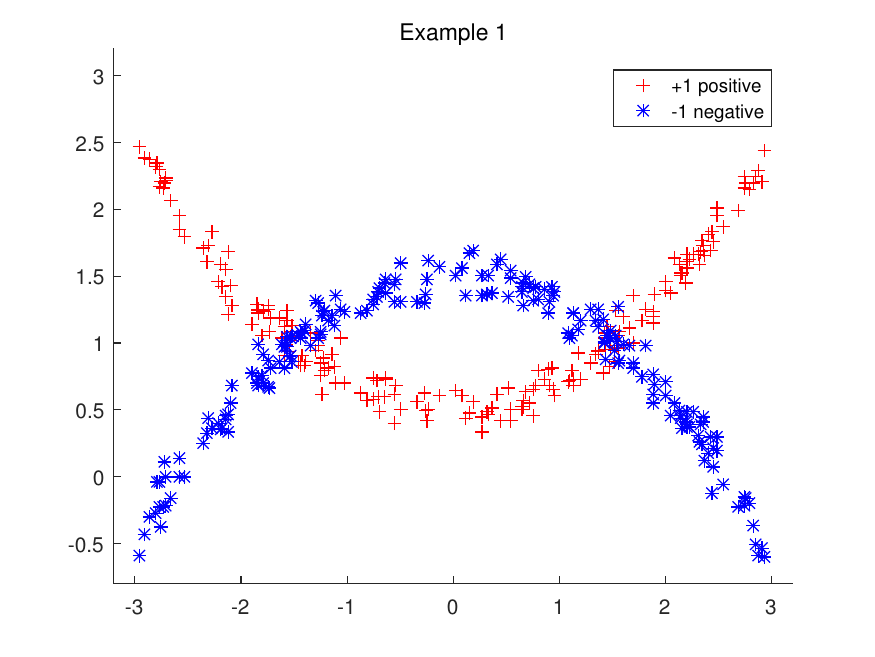}
\end{minipage}%
}%
\subfigure[Example 1(5$\%$)]{
\begin{minipage}[t]{0.3\linewidth}
\centering
\includegraphics[width=4.5cm]{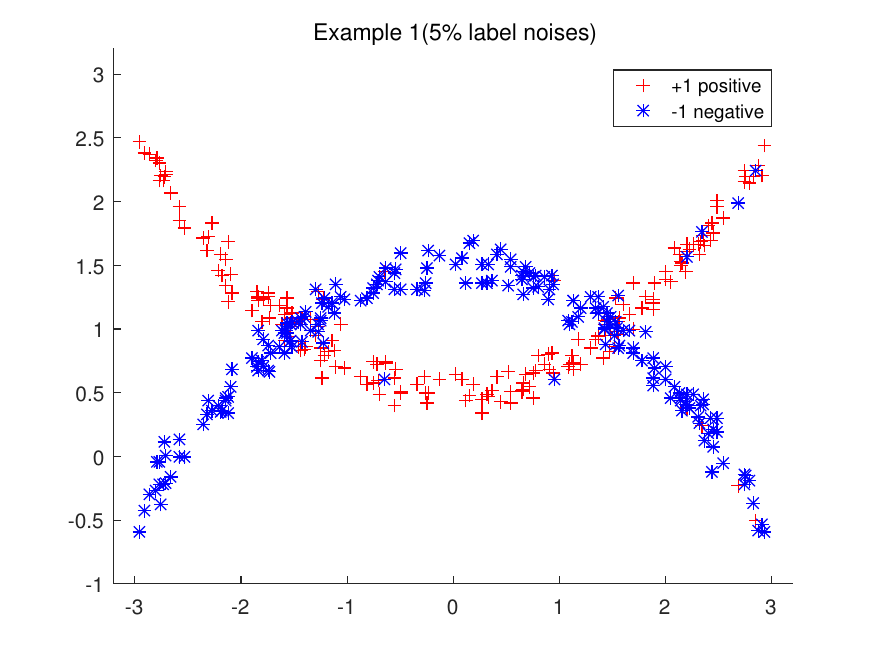}
\end{minipage}%
}
\subfigure[Example 1(10$\%$)]{
\begin{minipage}[t]{0.3\linewidth}
\centering
\includegraphics[width=4.5cm]{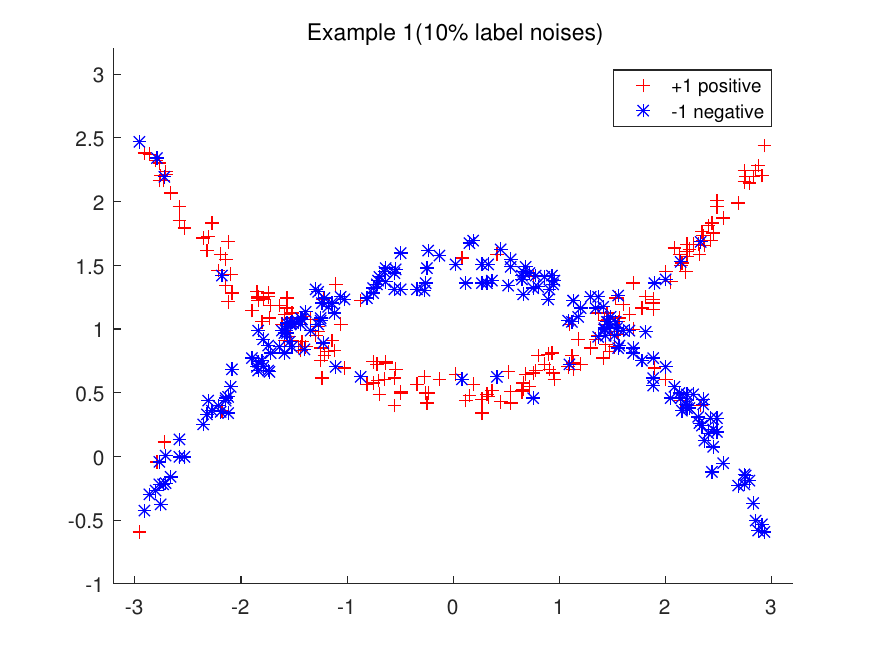}
\end{minipage}%
}
\centering
\caption{The distribution of data for Example 1 with different label noise ratios.}\label{example1}
\end{figure}

\textbf{Example 2}:\\
class +1:  $[x_i]_1=3cos(\theta_i),[x_i]_2=3sin(\theta_i)+\xi_i$,  $\theta_i \sim U[0,2\pi]$, \\
class -1:  $[x_j]_1=3cos(\theta_j),[x_j]_2=3sin(\theta_j)+\xi_j$, $\theta_j \sim U[\pi,2\pi]$,\\
where  $\xi_i,\xi_j \sim N(0,0.2)$,$i,j=1,2.... ,200$.

\begin{figure}[htbp]
\centering
\subfigure[Example 1]{
\begin{minipage}[t]{0.3\linewidth}
\centering
\includegraphics[width=4.5cm]{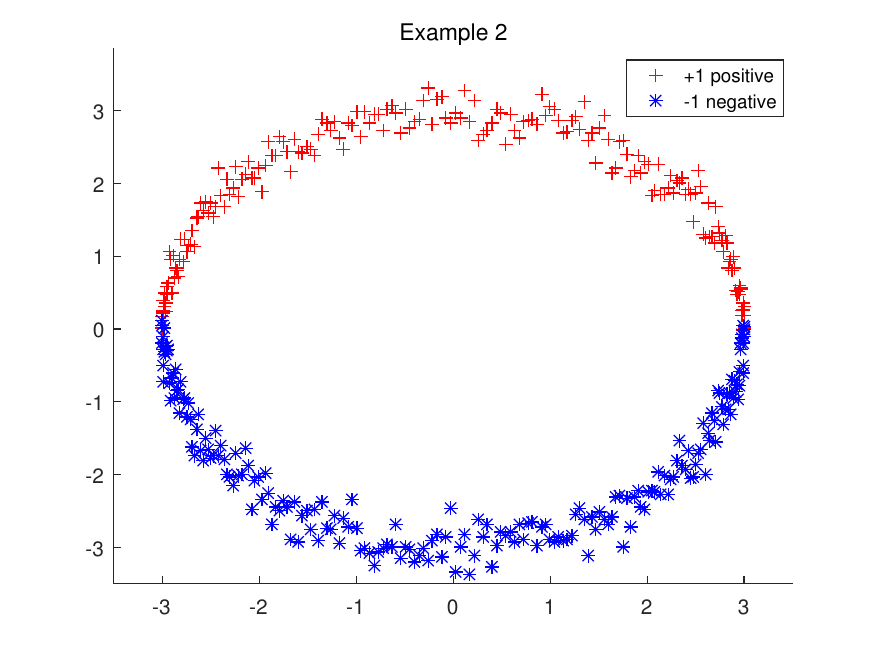}
\end{minipage}%
}%
\subfigure[Example 1(5$\%$)]{
\begin{minipage}[t]{0.3\linewidth}
\centering
\includegraphics[width=4.5cm]{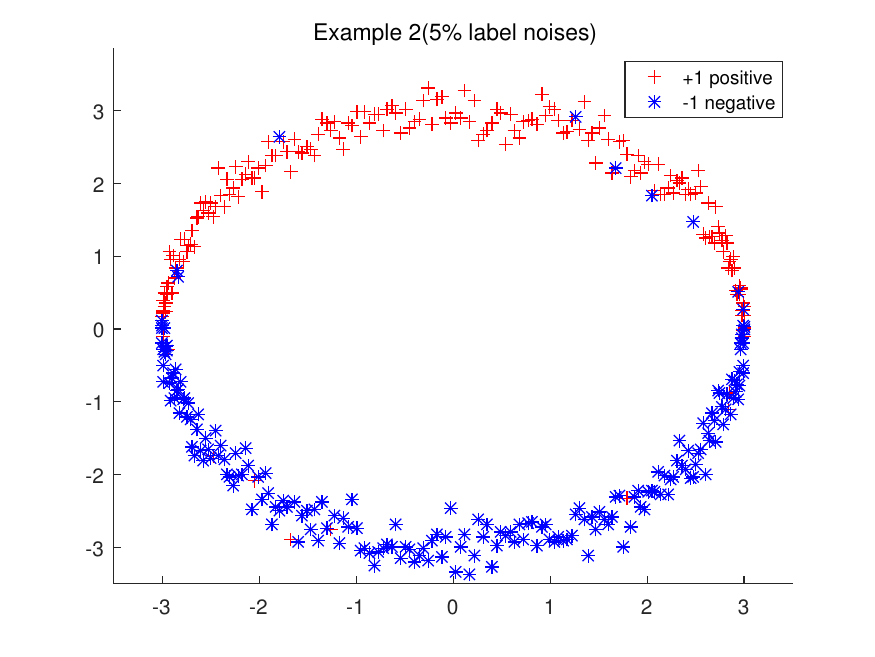}
\end{minipage}%
}
\subfigure[Example 1(10$\%$)]{
\begin{minipage}[t]{0.3\linewidth}
\centering
\includegraphics[width=4.5cm]{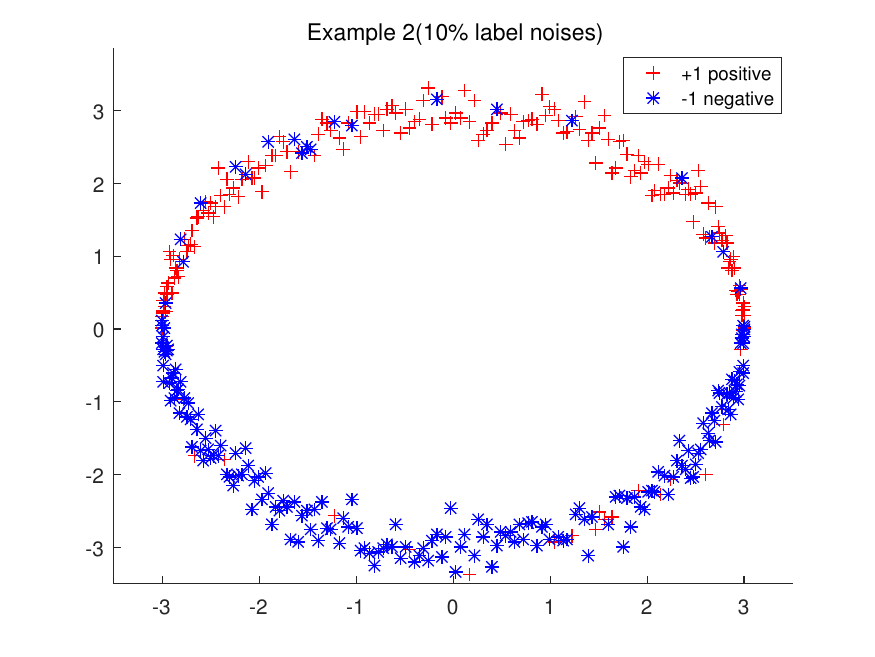}
\end{minipage}%
}
\centering
\caption{The distribution of data for Example 2 with different label noise ratios.}\label{example1}
\end{figure}

\textbf{Example 3}:\\
class +1:   $[x_i]_2=0.75[x_i]_1^2+1.5[x_i]_1+0.75+\xi_i$, $[x_i]_1 \sim U[-3,1]$, \\
class -1:  $[x_j]_2=0.75[x_j]_1^2-1.5[x_j]_1+0.75+\xi_j$, $[x_j]_1 \sim U[-1,3]$,\\
where  $\xi_i,\xi_j \sim N(0,0.1)$,$i,j=1,2.... ,200$.

\begin{figure}[htbp]
\centering
\subfigure[Example 1]{
\begin{minipage}[t]{0.3\linewidth}
\centering
\includegraphics[width=4.5cm]{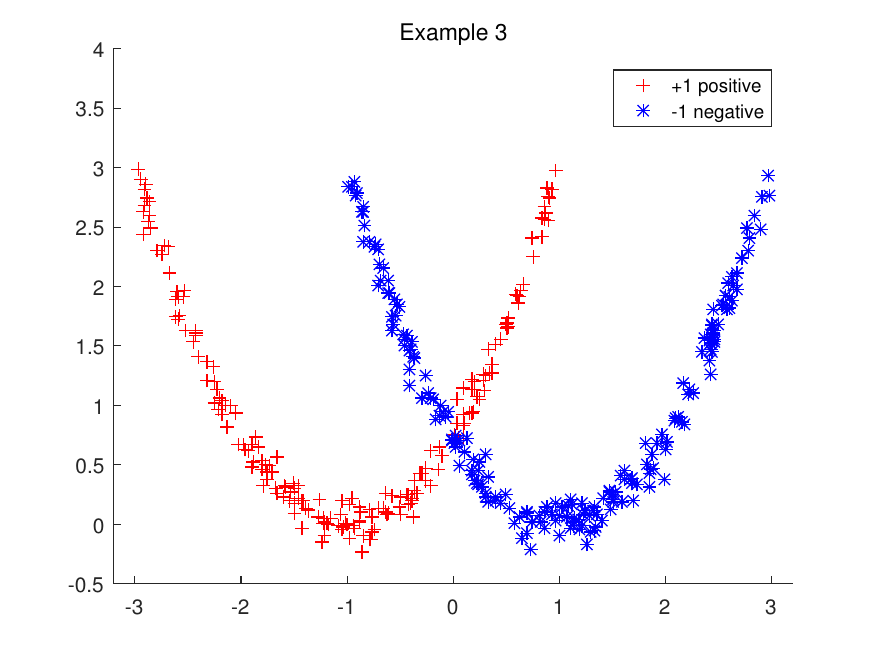}
\end{minipage}%
}%
\subfigure[Example 1(5$\%$)]{
\begin{minipage}[t]{0.3\linewidth}
\centering
\includegraphics[width=4.5cm]{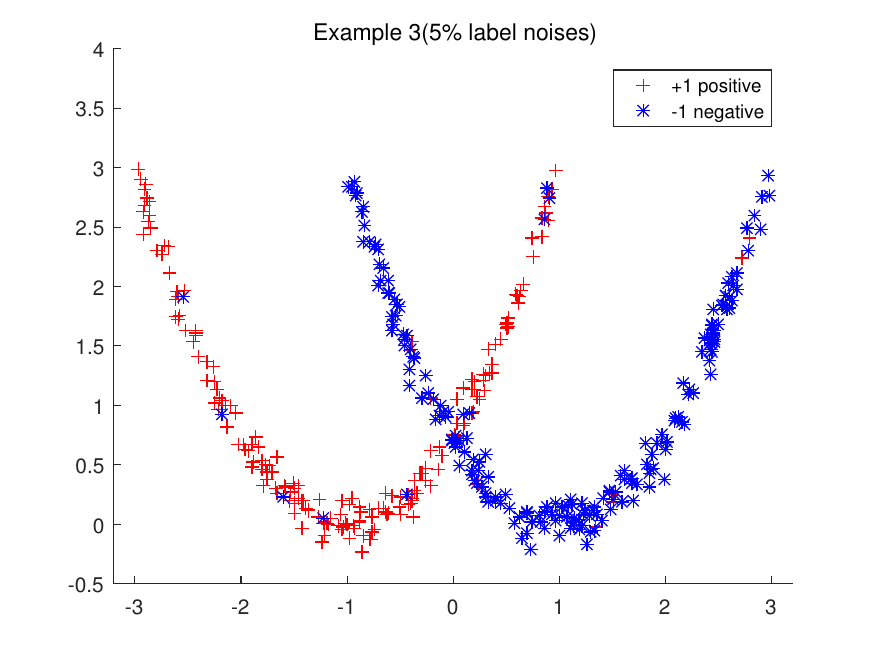}
\end{minipage}%
}
\subfigure[Example 1(10$\%$)]{
\begin{minipage}[t]{0.3\linewidth}
\centering
\includegraphics[width=4.5cm]{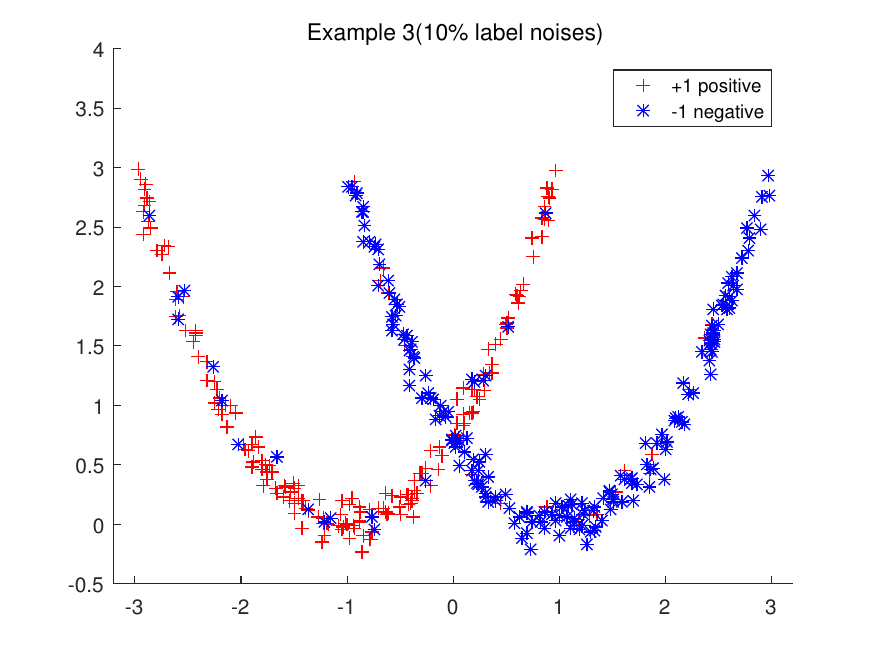}
\end{minipage}%
}
\centering
\caption{The distribution of data for Example 3 with different label noise ratios.}\label{example1}
\end{figure}

The classification results on the synthetic datasets as well as on the synthetic datasets with different percentages of labeled noises are shown in Table \ref{rengongjieguo}. It should be noted that for the FRTELM method, the activation function is $\frac{1}{1+exp(-(\boldsymbol{w}^{\top}\boldsymbol{x}+b))}$, where $\boldsymbol{w}$ and $b$ are randomly generated. From the classification results in Table \ref{rengongjieguo}, it can be seen that on the synthetic datasets Example 1, Example 2 and Example 3 without label noises, the classification accuracies of our C$L_1$QTSVM model are not the highest, but second to the optimal method. However, our C$L_1$QTSVM model has the highest classification accuracy on the synthetic datasets Example 1, Example 2 and Example 3 with label noises. Specifically, on the dataset Example 1 with 5$\%$ and 10$\%$ label noises, the classification accuracies of our C$L_1$QTSVM model are higher than the second ranked method by 0.71$\%$ and 1.83$\%$, respectively. And on the dataset Example 2 with 5$\%$ and 10$\%$ label noises, the classification accuracies of our C$L_1$QTSVM model are higher than the second ranked method by 0.15$\%$ and 0.08$\%$, respectively.  And on the dataset Example 3 with 5$\%$ and 10$\%$ label noises, the classification accuracy of our C$L_1$QTSVM model is higher than the second-ranked method by  0.62$\%$ and 1.80$\%$, respectively. These results validate the advantage of our method in handling data with label noises, which also prepare the way for the experiments of our C$L_1$QTSVM model on real datasets.

\renewcommand\tabcolsep{5pt}
\begin{table}[htbp]
\centering
\caption{Experimental results on synthetic datasets.}\label{rengongjieguo}
\vspace{-0.3cm}
\resizebox{.95\columnwidth}{!}{
\begin{tabular}{lllllllll} 
\toprule
Datasets&TSVM&C$L_1$FRTBSVM&C$L_1$FRTELM&LINEX-TSVM&C$L_{2,p}$-LSTSVM&$\nu$-FRSQSSVM&LSQTSVM&C$L_1$QTSVM\\
   &(Acc,Std)($\%$)&(Acc,Std)($\%$)&(Acc,Std)($\%$)& (Acc,Std)($\%$)& (Acc,Std)($\%$)&(Acc,Std)($\%$)& (Acc,Std)($\%$)&(Acc,Std)($\%$)\\
    &($c_2$,$\sigma$)&($c_1$,$c_2$,$\sigma$) &($c_2$,$L$)&($c_1$,$c_2$)&($c_2$,$\epsilon$,$p$)&$(\nu)$&($c_2$)&($c_1$,$c_2$)\\
    \midrule
  \multirow{2}{*}{Example 1} &(91.75,1.13)&(\textbf{94.45},\textbf{0.16})&(91.53,1.00)&(53.88,0.39)&(62.86,0.56)&(50.40,1.10)&(93.95,0.31)&(94.03,0.38)\\
                                &($10^{-5}$,$2^{-1}$)&($10^{3}$,$10^{-5}$,$2^{5}$)&($10^{-3}$,20)&($10^{2}$,$10^{-4}$)&($10^{-5}$,$10^{-5}$,1.7)&(0.01)&($10^{-1}$)&($10^{2}$,$10^{1}$)\\
    \multirow{2}{*}{Example 1(5$\%$)} &(86.77,0.77)&(85.88,0.82)&(83.60,1.37)&(51.96,0.91)&(59.36,0.34)&(50.52,1.12)&(83.25,0.50)&(\textbf{87.48},\textbf{0.36})\\
                                &($10^{-5}$,$2^{-3}$)&($10^{5}$,$10^{-5}$,$2^{0}$)&($10^{-5}$,200)&($10^{-3}$,$10^{4}$)&($10^{-3}$,$10^{-3}$,1.7)&(0.01)&($10^{0}$)&($10^{-1}$,$10^{2}$)\\
\multirow{2}{*}{Example 1(10$\%$)} &(80.88,1.17)&(80.92,1.21)&(77.10,1.89)&(53.27,2.39)&(59.92,0.58)&(46.50,1.92)&(64.20,1.12)&(\textbf{82.75},\textbf{0.24})\\
                                &($10^{-5}$,$2^{-3}$)&($10^{5}$,$10^{5}$,$2^{0}$)&($10^{-5}$,500)&($10^{-2}$,$10^{-4}$)&($10^{-3}$,$10^{-1}$,1.7)&(0.4)&($10^{0}$)&($10^{0}$,$10^{4}$)\\
\multirow{2}{*}{Example 2} &(\textbf{98.50},\textbf{0.20})&(98.05,0.42)&(97.55,0.63)&(89.02,0.57)&(98.07,0.16)&(97.75,0.61)&(97.77,0.18)&(98.17,0.26)\\
                                &($10^{-5}$,$2^{-3}$)&($10^{1}$,$10^{-1}$,$2^{3}$)&($10^{-3}$,100)&($10^{-4}$,$10^{-1}$)&($10^{3}$,$10^{3}$,1.5)&(0.2)&($10^{-2}$)&($10^{0}$,$10^{2}$)\\
\multirow{2}{*}{Example 2(5$\%$)} &(93.42,0.61)&(93.35,0.21)&(92.85,1.74)&(85.56,0.54)&(93.52,0.29)&(93.32,0.29)&(89.75,1.03)&(\textbf{93.67},\textbf{0.29})\\
                                &($10^{-3}$,$2^{-1}$)&($10^{-1}$,$10^{-1}$,$2^{0}$)&($10^{-1}$,40)&($10^{-4}$,$10^{1}$)&($10^{1}$,$10^{5}$,1.3)&(0.4)&($10^{3}$)&($10^{-3}$,$10^{-5}$)\\
\multirow{2}{*}{Example 2(10$\%$)} &(87.35,1.25)&(87.92,0.42)&(85.82,0.88)&(78.90,0.42)&(87.20,0.25)&(86.95,0.39)&(82.70,0.39)&(\textbf{88.00},\textbf{1.60})\\
                                &($10^{-5}$,$2^{-1}$)&($10^{-1}$,$10^{-5}$,$2^{0}$)&($10^{-5}$,500)&($10^{0}$,$10^{3}$)&($10^{-3}$,$10^{3}$,1.3)&(0.4)&($10^{-5}$)&($10^{-2}$,$10^{-1}$)\\
\multirow{2}{*}{Example 3} &(98.50,0.42)&(98.63,0.56)&(98.07,0.64)&(86.19,0.29)&(82.47,0.14)&(74.42,0.70)&(\textbf{98.95},\textbf{0.11})&(98.83,0.12)\\
                                &($10^{-5}$,$2^{-3}$)&($10^{-3}$,$10^{-1}$,$2^{0}$)&($10^{5}$,500)&($10^{0}$,$10^{-3}$)&($10^{-5}$,$10^{3}$,1.5)&(0.01)&($10^{-5}$)&($10^{-5}$,$10^{-4}$)\\
\multirow{2}{*}{Example 3(5$\%$)} &(91.78,0.59)&(92.30,0.39)&(90.50,0.74)&(82.24,0.55)&(79.07,0.88)&(74.72,0.55)&(87.78,1.11)&(\textbf{92.92},\textbf{1.81})\\
                                &($10^{-5}$,$2^{-2}$)&($10^{5}$,$10^{-3}$,$2^{0}$)&($10^{-5}$,500)&($10^{-2}$,$10^{-5}$)&($10^{-5}$,$10^{-1}$,1.5)&(0.2)&($10^{0}$)&($10^{-3}$,$10^{5}$)\\
\multirow{2}{*}{Example 3(10$\%$)} &(86.20,1.02)&(86.15,0.46)&(84.38,1.63)&(76.83,0.70)&(75.86,0.41)&(70.42,0.54)&(77.03,0.85)&(\textbf{88.00},\textbf{1.60})\\
                                &($10^{-4}$,$2^{-1}$)&($10^{-1}$,$10^{-5}$,$2^{0}$)&($10^{-1}$,20)&($10^{-1}$,$10^{-4}$)&($10^{-5}$,$10^{3}$,1.5)&(0.4)&($10^{2}$)&($10^{-1}$,$10^{3}$)\\
  \bottomrule
\end{tabular}
}
\end{table}

\subsubsection{Performance on benchmark datasets}
Table \ref{jizhun11} shows the classification results of our C$L_1$QTSVM model with the rest of the compared methods on the 16 benchmark datasets, where the optimal results are shown in bold. In addition, for the FRTELM method, the activation function is $\frac{1}{1+exp(-(\boldsymbol{w}^{\top}\boldsymbol{x}+b))}$, where $\boldsymbol{w}$ and $b$ are randomly generated. From the experimental results on the 16 benchmark datasets without noise in Table \ref{jizhun11}, it can be seen that the classification accuracies of our C$L_1$QTSVM model are the highest on 8 of the datasets. Specifically, on the benchmark datasets Ionosphere, Sonar, Vote, WBC, Yeast (2 vs. 8), and Robotnavig, our C$L_1$QTSVM model outperforms the second-ranked method by 0.17$\%$, 0.62$\%$, 0.20$\%$, and 0.04$\%$, 0.13$\%$ and 0.03$\%$,respectively; Also, although our C$L_1$QTSVM model does not have the highest accuracy on the datasets Indian, Blood, and Glass (0-1-5 vs. 2), it is only lower than the classification accuracy of the best method on these three datasets by 0.67$\%$, 0.40$\%$, and 0.20$\%$, respectively. These results demonstrate the effectiveness of our method.

\renewcommand\tabcolsep{5pt}
\begin{table}[htbp]
\centering
\caption{Experimental results on the benchmark dataset.}\label{jizhun11}
\vspace{-0.3cm}
\resizebox{.95\columnwidth}{!}{
\begin{tabular}{lllllllll} 
\toprule
  Datasets&TSVM&C$L_1$FRTBSVM&C$L_1$FRTELM&LINEX-TSVM&C$L_{2,p}$-LSTSVM&$\nu$-FRSQSSVM&LSQTSVM&C$L_1$QTSVM\\
   &(Acc,Std)($\%$)&(Acc,Std)($\%$)&(Acc,Std)($\%$)& (Acc,Std)($\%$)& (Acc,Std)($\%$)&(Acc,Std)($\%$)& (Acc,Std)($\%$)&(Acc,Std)($\%$)\\
    &($c_2$,$\sigma$)&($c_1$,$c_2$,$\sigma$) &($c_2$,$L$)&($c_1$,$c_2$)&($c_2$,$\epsilon$,$p$)&$(\nu)$&($c_2$)&($c_1$,$c_2$)\\
    \midrule
  \multirow{2}{*}{Australian} &(81.25,0.28)&(86.46,0.50)&(85.87,0.60)&(85.43,0.19)&(86.57,0.17)&(86.65,0.21)&(80.81,1.52)&(\textbf{86.66},\textbf{0.32})\\
                                &($10^{-5}$,$2^{-1}$)&($10^{3}$,$10^{-3}$,$2^{4}$)&($10^{1}$,50)&($10^{1}$,$10^{1}$)&($10^{1}$,$10^{5}$,1.5)&(0.01)&($10^{-2}$)&($10^{-3}$,$10^{-5}$)\\
    \multirow{2}{*}{Indian} &(71.33,0.19)&(71.77,0.88)&(70.11,1.04)&(71.30,0.13)&(\textbf{72.71},\textbf{1.03})&(69.72,0.80)&(49.95,0.68)&(72.04,0.52)\\
                                &($10^{1}$,$2^{-4}$)&($10^{-3}$,$10^{-3}$,$2^{0}$)&($10^{-5}$,500)&($10^{-1}$,$10^{-5}$)&($10^{-3}$,$10^{-1}$,1.3)&(0.01)&($10^{-4}$)&($10^{-5}$,$10^{0}$)\\
\multirow{2}{*}{Blood} &(44.17,6.70)&(78.20,0.40)&(78.92,0.51)&(76.83,0.34)&(77.33,0.08)&(\textbf{78.64},\textbf{0.68})&(61.35,1.34)&(78.24,0.12)\\
                                &($10^{3}$,$2^{-1}$)&($10^{-5}$,$10^{-5}$,$2^{0}$)&($10^{1}$,200)&($10^{2}$,$10^{-2}$)&($10^{-1}$,$10^{5}$,1.5)&(0.2)&($10^{-4}$)&($10^{-5}$,$10^{0}$)\\
\multirow{2}{*}{Cleve} &(75.98,2.85)&(81.06,0.97)&(78.49,1.83)&(81.66,0.51)&(82.19,0.69)&(\textbf{83.44},\textbf{0.18})&(67.28,2.17)&(81.07,1.52)\\
                                &($10^{5}$,$2^{1}$)&($10^{1}$,$10^{-5}$,$2^{0}$)&($10^{-5}$,20)&($10^{-5}$,$10^{4}$)&($10^{-3}$,$10^{3}$,1.3)&(0.01)&($10^{-1}$)&($10^{-5}$,$10^{-2}$)\\
\multirow{2}{*}{Climate-simulation}&(93.72,0.42)&(92.74,0.23)&(\textbf{95.22},\textbf{0.51})&(85.60,0.69)&(95.01,0.50)&(94.44,0.41)&(93.20,0.59)&(92.94,0.53)\\
                                &($10^{-5}$,$2^{1}$)&($10^{1}$,$10^{-5}$,$2^{0}$)&($10^{-1}$,200)&($10^{-2}$,$10^{-4}$)&($10^{-1}$,$10^{3}$,1.5)&(0.01)&($10^{5}$)&($10^{-4}$,$10^{-4}$)\\
\multirow{2}{*}{CMC} &(58.72,1.29)&(70.47,0.46)&(68.10,0.55)&(68.39,0.17)&(67.55,0.22)&(67.94,0.61)&(69.61,0.34)&(\textbf{70.67},\textbf{0.40})\\
                                &($10^{5}$,$2^{-1}$)&($10^{3}$,$10^{-5}$,$2^{3}$)&($10^{-2}$,50)&($10^{-3}$,$10^{5}$)&($10^{-3}$,$10^{3}$,1.9)&(0.01)&($10^{-4}$)&($10^{-5}$,$10^{0}$)\\
\multirow{2}{*}{Cylinder-bands} &(72.18,1.34)&(68.40,3.92)&(68.77,2.63)&(67.94,0.75)&(70.18,0.64)&(\textbf{73.81},\textbf{0.76})&(69.18,1.44)&(72.99,1.74)\\
                                &($10^{-5}$,$2^{-3}$)&($10^{-3}$,$10^{3}$,$2^{0}$)&($10^{-5}$,50)&($10^{-5}$,$10^{5}$)&($10^{-1}$,$10^{3}$,1.5)&(0.01)&($10^{0}$)&($10^{-5}$,$10^{0}$)\\
\multirow{2}{*}{Glass(0-1-5 vs. 2)} &(91.13,0.05)&(90.91,0.79)&(78.88,9.25)&(58.12,0.11)&(90.16,0.34)&(\textbf{91.13},\textbf{0.04})&(73.82,2.26)&(90.93,0.28)\\
                                &($10^{-5}$,$2^{-2}$)&($10^{3}$,$10^{-1}$,$2^{0}$)&($10^{4}$,100)&($10^{-4}$,$10^{1}$)&($10^{-5}$,$10^{-5}$,1.9)&(0.8)&($10^{-1}$)&($10^{-3}$,$10^{2}$)\\
\multirow{2}{*}{Haberman} &(59.12,4.02)&(74.48,0.95)&(70.96,2.15)&(75.46,0.41)&(\textbf{76.17},\textbf{0.53})&(74.23,0.97)&(71.30,0.83)&(74.93,0.35)\\
                                &($10^{-3}$,$2^{-1}$)&($10^{5}$,$10^{-3}$,$2^{3}$)&($10^{1}$,50)&($10^{-2}$,$10^{5}$)&($10^{-3}$,$10^{3}$,1.5)&(0.01)&($10^{4}$)&($10^{-5}$,$10^{0}$)\\
\multirow{2}{*}{Hepatitis} &(81.00,1.22)&(80.66,1.62)&(80.43,1.58)&(76.77,0.81)&(82.10,1.15)&(\textbf{84.91},\textbf{1.32})&(73.41,2.37)&(83.38,1.48)\\
                                &($10^{3}$,$2^{1}$)&($10^{3}$,$10^{-5}$,$2^{3}$)&($10^{1}$,20)&($10^{-3}$,$10^{3}$)&($10^{-3}$,$10^{3}$,1.9)&(0.01)&($10^{-2}$)&($10^{4}$,$10^{5}$)\\
\multirow{2}{*}{Ionosphere} &(89.81,0.79)&(90.23,0.62)&(89.83,0.57)&(89.06,0.46)&(88.67,0.30)&(89.92,0.43)&(88.69,0.87)&(\textbf{90.40},\textbf{0.69})\\
                                &($10^{5}$,$2^{1}$)&($10^{1}$,$10^{-5}$,$2^{3}$)&($10^{1}$,20)&($10^{-2}$,$10^{3}$)&($10^{-1}$,$10^{-3}$,1.7)&(0.01)&($10^{-5}$)&($10^{1}$,$10^{4}$)\\
\multirow{2}{*}{Sonar} &(83.99,1.65)&(83.22,1.29)&(55.67,3.34)&(73.12,1.07)&(73.90,0.90)&(82.69,1.01)&(86.40,1.49)&(\textbf{87.02},\textbf{0.94})\\
                                &($10^{-2}$,$2^{2}$)&($10^{5}$,$10^{-3}$,$2^{3}$)&($10^{-2}$,40)&($10^{0}$,$10^{3}$)&($10^{-1}$,$10^{3}$,1.3)&(0.2)&($10^{-5}$)&($10^{-1}$,$10^{1}$)\\
\multirow{2}{*}{Vote} &(94.46,0.54)&(94.57,0.25)&(89.86,0.94)&(94.83,0.22)&(94.21,0.21)&(94.20,0.41)&(90.32,0.93)&(\textbf{95.03},\textbf{0.36})\\
                                &($10^{-5}$,$2^{1}$)&($10^{1}$,$10^{-5}$,$2^{3}$)&($10^{-2}$,20)&($10^{0}$,$10^{4}$)&($10^{-3}$,$10^{3}$,1.7)&(0.01)&($10^{3}$)&($10^{-5}$,$10^{-1}$)\\
\multirow{2}{*}{WBC} &(95.77,0.30)&(97.09,0.24)&(93.63,1.75)&(95.96,0.43)&(95.76,0.09)&(96.47,0.11)&(76.10,0.62)&(\textbf{97.13},\textbf{0.25})\\
                                &($10^{1}$,$2^{-1}$)&($10^{1}$,$10^{-5}$,$2^{0}$)&($10^{1}$,50)&($10^{-3}$,$10^{4}$)&($10^{-3}$,$10^{1}$,1.9)&(0.01)&($10^{1}$)&($10^{-1}$,$10^{-4}$)\\
\multirow{2}{*}{Yeast(2 vs. 8)} &(94.33,0.39)&(94.47,0.71)&(92.32,4.59)&(95.21,0.65)&(95.16,0.21)&(95.20,0.17)&(95.08,0.26)&(\textbf{95.34},\textbf{0.26})\\
                                &($10^{-1}$,$2^{-4}$)&($10^{-1}$,$10^{-5}$,$2^{0}$)&($10^{4}$,20)&($10^{-2}$,$10^{-2}$)&($10^{-5}$,$10^{3}$,1.3)&(0.01)&($10^{-3}$)&($10^{-3}$,$10^{-5}$)\\
\multirow{2}{*}{Robotnavig} &(94.19,0.26)&(98.78,0.08)&(94.27,0.33)&(80.78,0.32)&(84.83,0.11)&(88.17,1.46)&(95.08,0.26)&(\textbf{98.81},\textbf{0.09})\\
                                &($10^{-3}$,$2^{-1}$)&($10^{1}$,$10^{-5}$,$2^{3}$)&($10^{-2}$,500)&($10^{2}$,$10^{-2}$)&($10^{-1}$,$10^{1}$,1.3)&(0.01)&($10^{-2}$)&($10^{-5}$,$10^{-1}$)\\
  \bottomrule
\end{tabular}
}
\end{table}

\subsubsection{Performance on image datasets}
Table \ref{tuxiang1} shows the basic information of the four image datasets used in this subsection. And some images from these four image datasets are shown in Figure \ref{tuxiangshuju}(a)-(d). In particular, the USPS dataset and the COIL20 dataset are from the URL \footnote{http://www.cad.zju.edu.cn/home/dengcai/Data/MLData.html}, while the Yale dataset and the ORL dataset are from the URL \footnote{https:// jundongl.github.io/scikit-feature/datasets.html}. For the USPS handwritten digit recognition datasets, we use 9298 images, in which each image has 16$\times$16 pixels, and the label of each image belongs to one of the digits 0 to 9. For the COIL20 dataset, it contains 20 objects with 72 images each, and each image with a pixel size of 32$\times$32. Each object in this dataset has its images taken 5 degrees apart as it rotates on a turntable. For the Yale dataset, which consists of 165 grayscale images in GIF format of 15 individuals, there are 11 images for each subject, and each image has a different facial expression or composition, e.g., center light, glasses, happy, left light, no glasses, normal, right light, sad, sleepy, surprised, and blinking. And the pixel size of each image is 32$\times$32. For the ORL dataset, the dataset consists of 40 different subjects, each subject has 10 different photos. For some subjects, the images were taken at different times, with different lighting, facial expressions (eyes open/close, smiling/not smiling) and facial details (glasses/no glasses). All images were taken against a dark, uniform background, and subjects were held in an upright, frontal position (with some lateral movement allowed). In addition, each image in the ORL dataset has 32$\times$32 pixels.
\begin{figure}[htbp]
\centering
\subfigure[USPS]{
\begin{minipage}[t]{0.5\linewidth}
\centering
\includegraphics[ height=1.5cm,width=6cm]{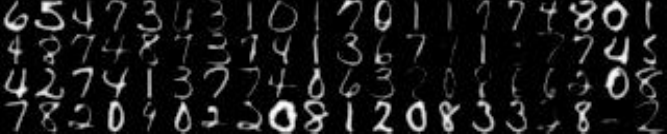}
\end{minipage}%
}%
\subfigure[ORL]{
\begin{minipage}[t]{0.5\linewidth}
\centering
\includegraphics[height=1.5cm, width=6cm]{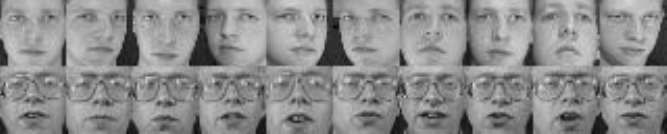}
\end{minipage}%
}%

\subfigure[Yale]{
\begin{minipage}[t]{0.5\linewidth}
\centering
\includegraphics[height=1.5cm, width=6cm]{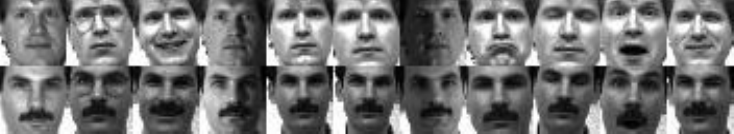}
\end{minipage}%
}%
\subfigure[COIL20]{
\begin{minipage}[t]{0.5\linewidth}
\centering
\includegraphics[height=1.5cm, width=6cm]{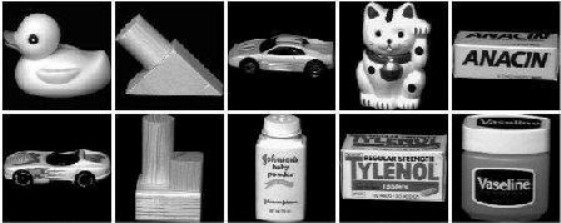}
\end{minipage}%
}%

		\caption{Visualization results for four image datasets.}\label{tuxiangshuju}
\end{figure}

Before we analyze the experimental results on the image dataset, one detail of the experiment needs to be explained. Since the feature dimensions of the image dataset are much higher than those of the benchmark datasets and the synthetic datasets, the time complexity of the C$L_1$QTSVM and LSQTSVM models grows quadratically with the feature dimensions. In order to reduce the time complexity and computational memory, in this paper, both the C$L_1$QTSVM model and LSQTSVM model are simplified using the diagonal vectorization operator and the reduced quadratic vectorization operator without cross terms \citep{gaonu-FRSQSSVM}, then the equations (\ref{Cgongshi1})-(\ref{Cgongshi4}) are rewritten as
\begin{equation}\label{RCgongshi1}
   \begin{split}
    \boldsymbol{\omega}_{+}\triangleq\left[\operatorname{dvec}^{\top}(\boldsymbol{W}_{+}), \boldsymbol{b}_{+}^{\top}, c_+\right]^{\top}\in \mathbb{R}^{2n+1},
    \end{split}
\end{equation}
\begin{equation}\label{RCgongshi2}
   \begin{split}
    \boldsymbol{\omega}_{-}\triangleq\left[\operatorname{dvec}^{\top}(\boldsymbol{W}_{-}), \boldsymbol{b}_{-}^{\top},c_+\right]^{\top}\in \mathbb{R}^{2n+1},
    \end{split}
\end{equation}
\begin{equation}\label{RCgongshi3}
   \begin{split}
   \boldsymbol{h}_{j}^{-}\triangleq\left[\operatorname{qvec}^{\top}(\boldsymbol{x}_{j}), \boldsymbol{x}_{j}^{\top}, 1\right]^{\top}\in \mathbb{R}^{2n+1},j=1, 2,\ldots, m_{-},
    \end{split}
\end{equation}
\begin{equation}\label{RCgongshi4}
   \begin{split}
   \boldsymbol{h}_{i}^{+}\triangleq\left[\operatorname{qvec}^{\top}(\boldsymbol{x}_{i}), \boldsymbol{x}_{i}^{\top}, 1\right]^{\top},\in \mathbb{R}^{2n+1},i=1, 2,\ldots, m_{+}.
    \end{split}
\end{equation}

Therefore, based on the Eqs. (\ref{RCgongshi1})-(\ref{RCgongshi4}), the time complexity of the C$L_1$QTSVM model and the LSQTSVM model on the image datasets is greatly decreased.

The data are normalized to $[-1, 1]$ before performing the experiments. Both the C$L_1$FRTBSVM model and the TSVM model use the RBF kernel function. The evaluation metrics used are Acc and F1, which are obtained by the average of 10 times 10-fold cross-validation. Time is the sum of the training time and testing time of a 10-fold cross validation. Table \ref{tuxiang11} and Table \ref{tuxiang2} show the classification results and optimal parameters of the C$L_1$QTSVM model and the rest of the compared state-of-the-art methods on the four image datasets. As shown in Table \ref{tuxiang11}, the classification accuracy of our C$L_1$QTSVM model on Yale dataset, USPS dataset and ORL dataset are the optimal results, which are higher than the second ranked method by $1.94$\%$, $0.01$\%$ and 0.05$\%$ respectively. The C$L_1$QTSVM model is not the optimal result on the COIL20 dataset, but its classification accuracy is only 0.45$\%$ lower than that of the C$L_1$FRTBSVM model. Similarly, from Table \ref{tuxiang2}, it can also be seen that the F1 scores of  C$L_1$QTSVM model on the Yale dataset, USPS dataset, and ORL dataset are the optimal results, which are higher than the second-ranked method by 2.37$\%$, 0.62$\%$, and 0.08$\%$, respectively. The C$L_1$QTSVM model is not the optimal result on the COIL20 dataset, but its classification accuracy is only 0.35$\%$ lower than that of the C$L_1$FRTBSVM model. 

Table \ref{tuxiang3} briefly shows the running time of each method on the four image datasets, where the running time is the sum of the testing time and training time. From the results shown in Table \ref{tuxiang3}, the following conclusions about our C$L_1$QTSVM model can be obtained: (1) The running time of the C$L_1$QTSVM model on the USPS dataset is much lower than that of the TSVM, C$L_1$FRTBSVM, and $\nu$-FRSQSSVM models, which is mainly due to the fact that our model uses the SMW theorem and the reduced quadratic quantization operator without cross-terms. (2) On the COIL20, USPS and ORL datasets, the running time of the C$L_1$QTSVM model is not as good as that of the C$L_1$FRTBSVM, the C$L_1$FRTELM and the LINEX-TSVM, but it is not much different from them. (3) As shown from the average running time results, the running time of the C$L_1$QTSVM model is not only much lower than that of TSVM, C$L_1$FRTELM, $\nu$-FRSQSSVM and C$L_{2,p}$-LSTSVM, but also comparable to that of C$L_1$FRTELM, LINEX-TSVM and LSQTSVM. .

\renewcommand\tabcolsep{30pt}
\begin{table*}[htbp]
\centering
\caption{Basic information on four image recognition datasets.}\label{tuxiang1}
\vspace{-0.3cm}
\resizebox{.99\columnwidth}{!}{
\begin{tabular}{llllll} 
\hline
 Dataset & Classification &\#Instances&\#Attributes \\ \hline
   COIL20&(0$\sim$10,11$\sim$20)&1440&1024\\
		                       USPS&(0$\sim$4,5$\sim$9)&9298&256\\
		                       Yale&(1$\sim$7,8$\sim$15)&165&1024\\
		                       ORL&(1$\sim$5,6$\sim$10)&100&1024\\
		
  \hline
\end{tabular}
}
\end{table*}

\renewcommand\tabcolsep{5pt}
\begin{table}[htbp]
\centering
\caption{Optimal accuracy and optimal parameter results for each method on four image datasets.}\label{tuxiang11}
\vspace{-0.3cm}
\resizebox{.99\columnwidth}{!}{
\begin{tabular}{lllllllll} 
\toprule
 Datasets&TSVM&C$L_1$FRTBSVM&C$L_1$FRTELM&LINEX-TSVM&C$L_{2,p}$-LSTSVM&$\nu$-FRSQSSVM&LSQTSVM&C$L_1$QTSVM\\
   &(Acc,Std)($\%$)&(Acc,Std)($\%$)&(Acc,Std)($\%$)& (Acc,Std)($\%$)& (Acc,Std)($\%$)&(Acc,Std)($\%$)& (Acc,Std)($\%$)&(Acc,Std)($\%$)\\
    &($c_2$,$\sigma$)&($c_1$,$c_2$,$\sigma$) &($c_2$,$L$)&($c_1$,$c_2$)&($c_2$,$\epsilon$,$p$)&$(\nu)$&($c_2$)&($c_1$,$c_2$)\\
    \midrule
  \multirow{2}{*}{COIL20} &(99.72,0.13)&(\textbf{100.00},\textbf{0.00})&(52.23,1.98)&(94.66,0.12)&(95.89,0.15)&(95.80,0.37)&(96.98,0.19)&(99.55,0.09)\\
                                &($10^{5}$,$2^{3}$)&($10^{1}$,$10^{-5}$,$2^{3}$)&($10^{1}$,50)&($10^{-1}$,$10^{5}$)&($10^{-3}$,$10^{3}$,1.9)&(0.01)&($10^{2}$)&($10^{-5}$,$10^{-2}$)\\
    \multirow{2}{*}{USPS} &(90.26,2.65)&(91.21,0.64)&(56.79,1.77)&(81.26,0.08)&(86.94,0.07)&(90.05,0.58)&(84.58,0.45)&(\textbf{91.22},\textbf{0.09})\\
                                &($10^{-3}$,$2^{1}$)&($10^{-5}$,$10^{-2}$,$2^{3}$)&($10^{-3}$,100)&($10^{-1}$,$10^{5}$)&($10^{-1}$,$10^{5}$,1.5)&(0.01)&($10^{-1}$)&($10^{-5}$,$10^{-1}$)\\
\multirow{2}{*}{Yale} &(82.30,1.56)&(82.09,0.82)&(54.03,3.52)&(76.42,0.85)&(80.13,0.44)&(83.31,1.45)&(78.11,1.06)&(\textbf{85.25},\textbf{1.66})\\
                                &($10^{5}$,$2^{3}$)&($10^{1}$,$10^{-5}$,$2^{3}$)&($10^{3}$,500)&($10^{-1}$,$10^{5}$)&($10^{-5}$,$10^{5}$,1.9)&(0.2)&($10^{2}$)&($10^{1}$,$10^{5}$)\\
\multirow{2}{*}{ORL} &(99.60,0.52)&(99.75,0.37)&(50.70,5.91)&(96.48,0.75)&(99.40,0.34)&(99.20,0.63)&(97.70,0.67)&(\textbf{99.80},\textbf{0.42})\\
                                &($10^{-3}$,$2^{3}$)&($10^{1}$,$10^{-5}$,$2^{3}$)&($10^{3}$,50)&($10^{-3}$,$10^{1}$)&($10^{-3}$,$10^{1}$,1.3)&(0.01)&($10^{-5}$)&($10^{-5}$,$10^{-2}$)\\
   \textbf{Avg.Acc}&92.97&93.26&53.44&87.21&90.59&92.09&89.34&\textbf{93.96}\\
  \textbf{Avg.rank}&2.75&2.25&8.00&7.00&4.75&4.25&5.50&\textbf{1.50}\\
  \bottomrule
\end{tabular}
}
\end{table}

\renewcommand\tabcolsep{5pt}
\begin{table}[htbp]
\centering
\caption{Optimal F1 scores and optimal parameter results for each method on four image datasets.}\label{tuxiang2}
\vspace{-0.3cm}
\resizebox{.99\columnwidth}{!}{
\begin{tabular}{lllllllll} 
\toprule
 Datasets&TSVM&C$L_1$FRTBSVM&C$L_1$FRTELM&LINEX-TSVM&C$L_{2,p}$-LSTSVM&$\nu$-FRSQSSVM&LSQTSVM&C$L_1$QTSVM\\
   &(F1,Std)($\%$)&(F1,Std)($\%$)&(F1,Std)($\%$)& (F1,Std)($\%$)& (F1,Std)($\%$)&(F1,Std)($\%$)& (F1,Std)($\%$)&(F1,Std)($\%$)\\
    &($c_2$,$\sigma$)&($c_1$,$c_2$,$\sigma$) &($c_2$,$L$)&($c_1$,$c_2$)&($c_2$,$\epsilon$,$p$)&$(\nu)$&($c_2$)&($c_1$,$c_2$)\\
    \midrule
  \multirow{2}{*}{COIL20} &(99.81,0.10)&(\textbf{100.00},\textbf{0.00})&(48.50,4.07)&(94.25,0.32)&(95.92,0.25)&(96.07,0.37)&(96.95,0.20)&(99.65,0.09)\\
                                &($10^{5}$,$2^{3}$)&($10^{1}$,$10^{-5}$,$2^{3}$)&($10^{-5}$,40)&($10^{-1}$,$10^{5}$)&($10^{-3}$,$10^{3}$,1.9)&(0.01)&($10^{2}$)&($10^{-5}$,$10^{-1}$)\\
    \multirow{2}{*}{USPS} &(87.83,2.98)&(88.74,1.05)&(NaN,NaN)&(75.04,0.07)&(83.89,0.08)&(87.93,0.77)&(80.41,0.81)&(\textbf{89.45},\textbf{0.12})\\
                                &($10^{-3}$,$2^{1}$)&($10^{-5}$,$10^{-2}$,$2^{3}$)&(-,-)&($10^{1}$,$10^{3}$)&($10^{-1}$,$10^{3}$,1.9)&(0.01)&($10^{-1}$)&($10^{-5}$,$10^{-1}$)\\
\multirow{2}{*}{Yale} &(79.05,2.03)&(81.08,2.27)&(NaN,NaN)&(79.19,0.69)&(71.99,1.58)&(81.56,3.31)&(77.38,1.19)&(\textbf{83.93},\textbf{2.22})\\
                                &($10^{3}$,$2^{3}$)&($10^{4}$,$10^{-5}$,$2^{3}$)&(-,-)&($10^{-1}$,$10^{5}$)&($10^{-5}$,$10^{5}$,1.9)&(0.2)&($10^{2}$)&($10^{1}$,$10^{5}$)\\
\multirow{2}{*}{ORL} &(99.78,0.49)&(99.86,0.45)&(NaN,NaN)&(97.62,0.42)&(98.09,0.58)&(99.11,0.71)&(97.30,1.34)&(\textbf{99.94},\textbf{0.19})\\
                                &($10^{-3}$,$2^{5}$)&($10^{1}$,$10^{-5}$,$2^{3}$)&(-,-)&($10^{1}$,$10^{5}$)&($10^{-5}$,$10^{3}$,1.3)&(0.01)&($10^{-5}$)&($10^{-5}$,$10^{-2}$)\\
    \textbf{Avg.F1}&91.62&92.40&NaN&85.53&87.47&91.17&88.01&\textbf{93.24}\\
  \textbf{Avg.rank}&3.50&2.00&8.00&6.00&5.75&3.50&5.75&\textbf{1.50}\\
  \bottomrule
\end{tabular}
}
\end{table}

\renewcommand\tabcolsep{5pt}
\begin{table}[htbp]
\centering
\caption{Running time of each method on four image datasets.}\label{tuxiang3}
\vspace{-0.3cm}
\resizebox{.99\columnwidth}{!}{
\begin{tabular}{lllllllll} 
\toprule
 Datasets&TSVM&C$L_1$FRTBSVM&C$L_1$FRTELM&LINEX-TSVM&C$L_{2,p}$-LSTSVM&$\nu$-FRSQSSVM&LSQTSVM&C$L_1$QTSVM\\
   &(Time)(s)&(Time)(s)&(Time)(s)&(Time)(s)&(Time)(s)&(Time)(s)&(Time)(s)&(Time)(s)\\
    \midrule
\multirow{1}{*}{COIL20} &(90.76)&(36.40)&(8.19)&(\textbf{8.11})&(161.22)&(1198.97)&(199.72)&(98.43)\\

    \multirow{1}{*}{USPS} &(1089.48)&(5363.28)&(69.73)&(181.67)&(467.56)&(11838.59)&(\textbf{65.08})&(100.83)\\

\multirow{1}{*}{Yale} &(0.62)&(\textbf{0.31})&(16.29)&(2.69)&(3.25)&(139.14)&(23.60)&(34.20)\\

\multirow{1}{*}{ORL} &(0.27)&(\textbf{0.21})&(1.11)&(2.57)&(3.93)&(88.65)&(14.65)&(35.31)\\
 \textbf{Avg.Time}&295.28&1350.05&\textbf{23.83}&48.76&158.99&3316.34&75.76&67.19\\      
  \bottomrule
\end{tabular}
}
\end{table}

\subsection{Robustness analysis}
In this subsection, to validate that our C$L_1$QTSVM model is insensitive to outliers on the benchmark datasets, we select six small-scale benchmark datasets for numerical experiments. As shown in Figure \ref{sanwei12}, the additive noise ratios are 0$\%$, 5$\%$ and 10$\%$, 15$\%$, 20$\%$, 25$\%$ and 30$\%$, respectively. As can be seen from the results, the classification performance of our C$L_1$QTSVM model is less affected by the increase of noise ratios. Specifically, on the datasets Australian, CMC, Sonar, WBC, Yeast (2 vs. 8) and Robotnavig, the classification performance of the C$L_1$QTSVM model slightly outperforms that of the rest of the compared methods. In particular, on the dataset Robotnavig, the robustness of the C$L_1$QTSVM model is higher than that of the remaining compared methods. However, on the CMC, Australian and WBC datasets, the classification performance of our C$L_1$QTSVM model is comparable to that of the C$L_1$FRTBSVM model. It should be noted that our C$L_1$QTSVM model has much better classification results on these six datasets compared with the traditional TSVM and LSQTSVM models. From the analysis of the above results, it can be seen that our C$L_1$QTSVM model has some advantages.

Our C$L_1$QTSVM model outperforms other state-of-the-art methods for the following reasons: (1) Capped $L_1$-norm is introduced into C$L_1$QTSVM model, which effectively reduces the influence of outliers on our model. (2) The introduction of $L_2$-regularization term improves the generalization ability of our model. (3) The designed iterative algorithm can have the ability to solve C$L_1$QTSVM model, which leads to better performance. (4) Our C$L_1$QTSVM model is kernel-free, which effectively solves the difficulty of selecting the kernel functions and kernel parameters of the comparative kernel methods on certain datasets.

\begin{figure}[htbp]
\centering
\subfigure[Australian]{
\begin{minipage}[t]{0.5\linewidth}
\centering
\includegraphics[width=6cm]{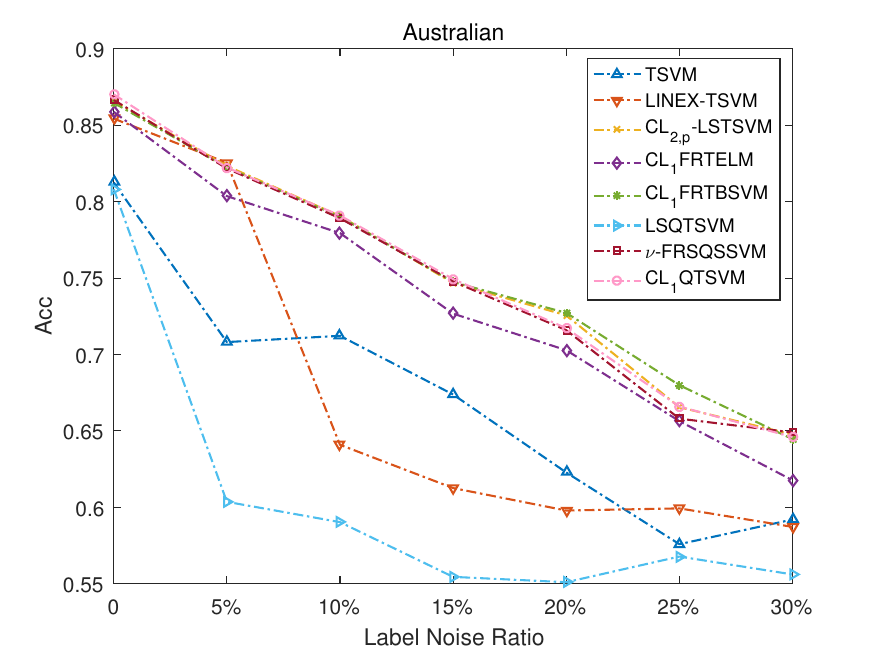}
\end{minipage}%
}%
\subfigure[CMC]{
\begin{minipage}[t]{0.5\linewidth}
\centering
\includegraphics[width=6cm]{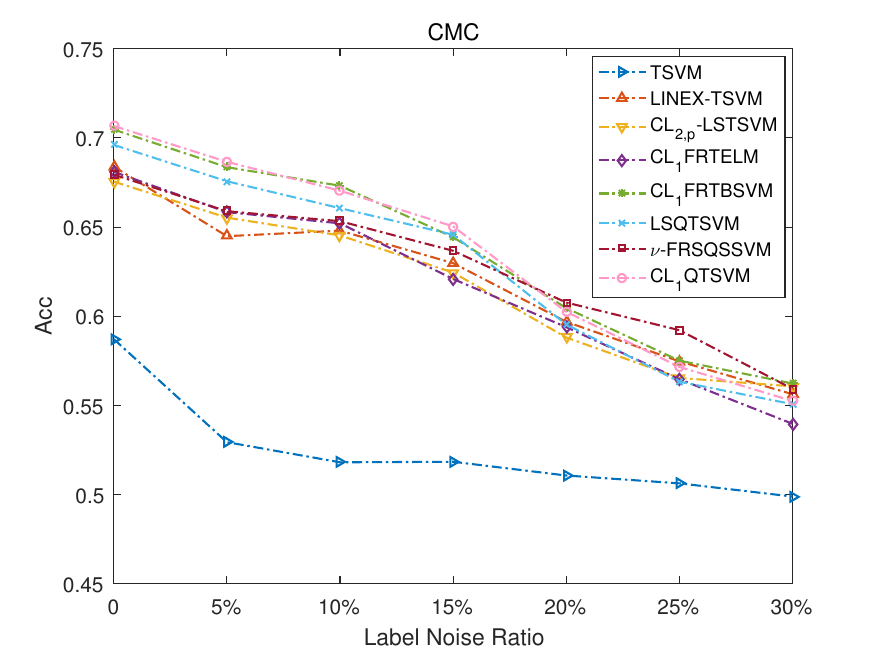}
\end{minipage}%
}%

\subfigure[Sonar]{
\begin{minipage}[t]{0.5\linewidth}
\centering
\includegraphics[width=6cm]{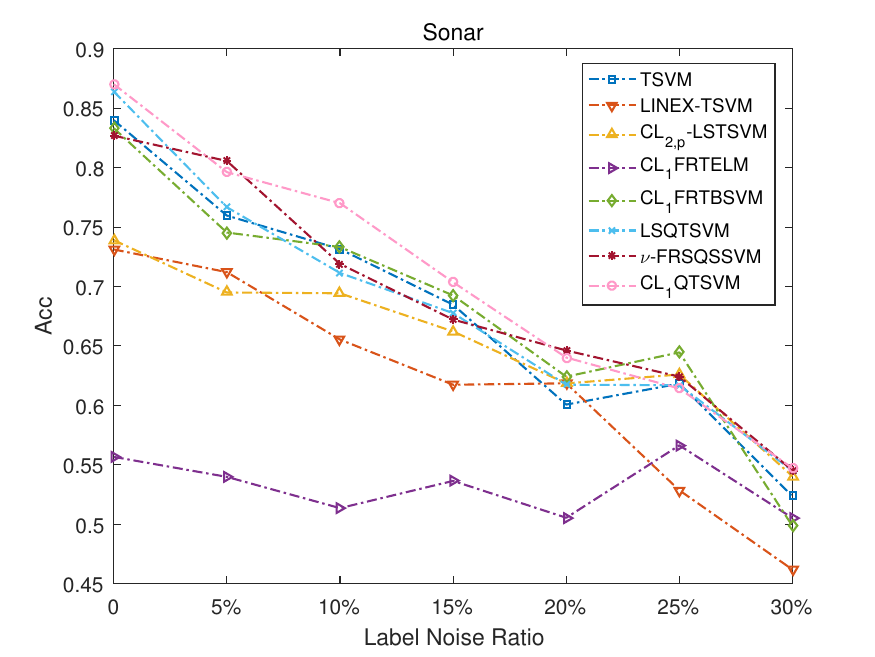}
\end{minipage}%
}%
\subfigure[WBC]{
\begin{minipage}[t]{0.5\linewidth}
\centering
\includegraphics[width=6cm]{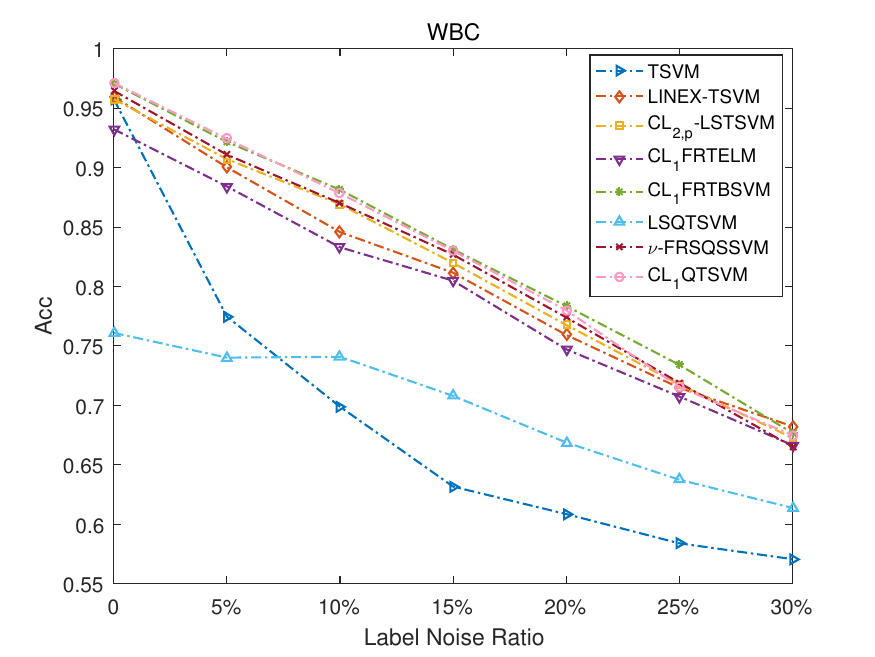}
\end{minipage}%
}%

\subfigure[Yeast(2 vs. 8)]{
\begin{minipage}[t]{0.5\linewidth}
\centering
\includegraphics[width=6cm]{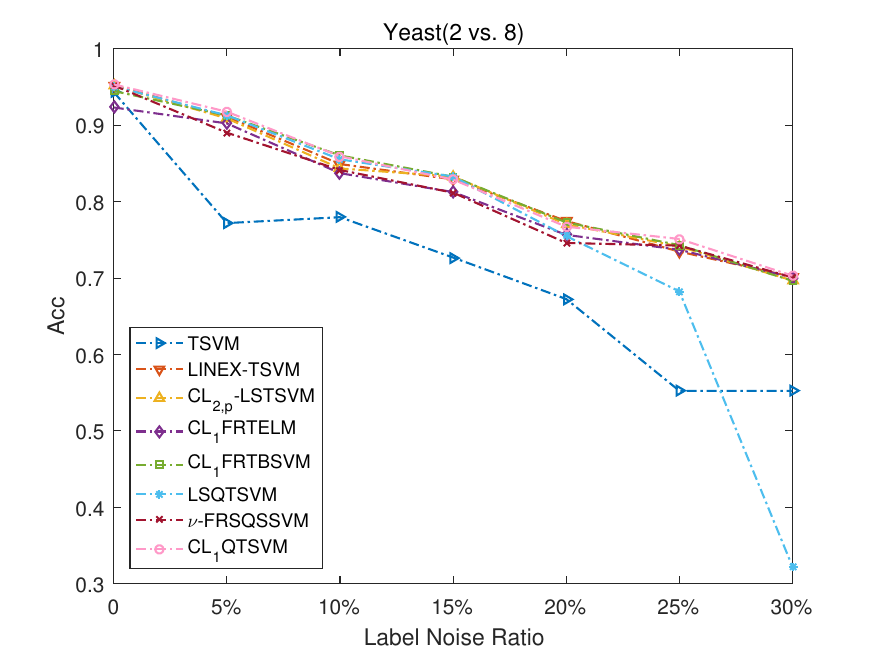}
\end{minipage}%
}%
\subfigure[Robotnavig]{
\begin{minipage}[t]{0.5\linewidth}
\centering
\includegraphics[width=6cm]{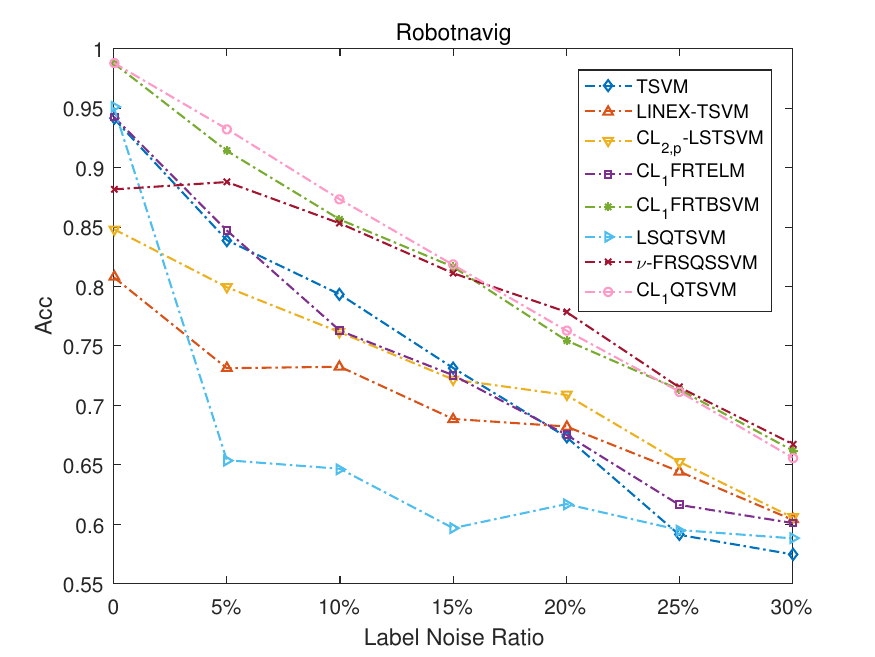}
\end{minipage}
}%
		\caption{Comparison of classification performance by methods with different label noises.}\label{sanwei12}
\end{figure}

\subsection{Parameter sensitivity analysis}
This subsection mainly analyzes the relationship between the regularization parameter $c_1$ and the penalty parameter $c_2$ of our C$L_1$ QTSVM model with the classification accuracy. Figure \ref{sanwei1} shows the results on six benchmark datasets. The results are analyzed as follows: on the dataset Australian, it can be seen that our C$L_1$QTSVM model is insensitive to the regularization parameter $c_1$, but is more sensitive to the penalty parameter $c_2$; when $c_2\in(10^{-5},10^{0})$, our C$L_1$QTSVM model can obtain the highest accuracy. On the dataset Climate-simulation, our C$L_1$QTSVM model is insensitive to both parameters $c_1$ and $c_2$. On the Ionosphere, our C$L_1$QTSVM model is insensitive to the parameter $c_1$; However, our C$L_1$QTSVM model can find the optimal parameter when $c_2\in(10^{-5},10^{-3})$. On the dataset Robotnavig, our C$L_1$QTSVM model finds the optimal parameter combinations when $c_2\in(10^{-5},10^{0})$ and $c_1\in(10^{-5},10^{-3})$. On the dataset Sonar and Yeast (2 vs. 8), our C$L_1$QTSVM model is insensitive to the parameter $c_1$; but when $c_2\in(10^{-5},10^{0})$, the C$L_1$QTSVM model finds the optimal parameter combination. Based on the above analysis, we can get the following conclusions: (1) Our C$L_1$QTSVM model is insensitive to the parameter $c_1$, which can be set to the range of $c_1\in(10^{-5},10^{0})$. (2) In order to obtain better classification performance, the optimal parameter range for parameter $c_2$ is $c_1\in(10^{-5},10^{0})$.

\begin{figure}[htbp]
\centering
\subfigure[Australian]{
\begin{minipage}[t]{0.5\linewidth}
\centering
\includegraphics[width=6cm]{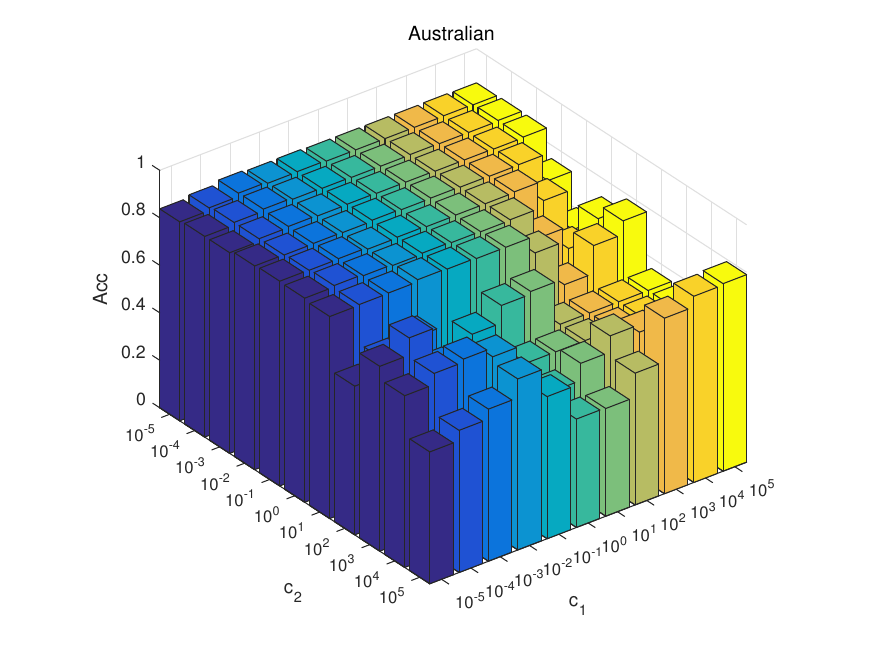}
\end{minipage}%
}%
\subfigure[Climate-simulation]{
\begin{minipage}[t]{0.5\linewidth}
\centering
\includegraphics[width=6cm]{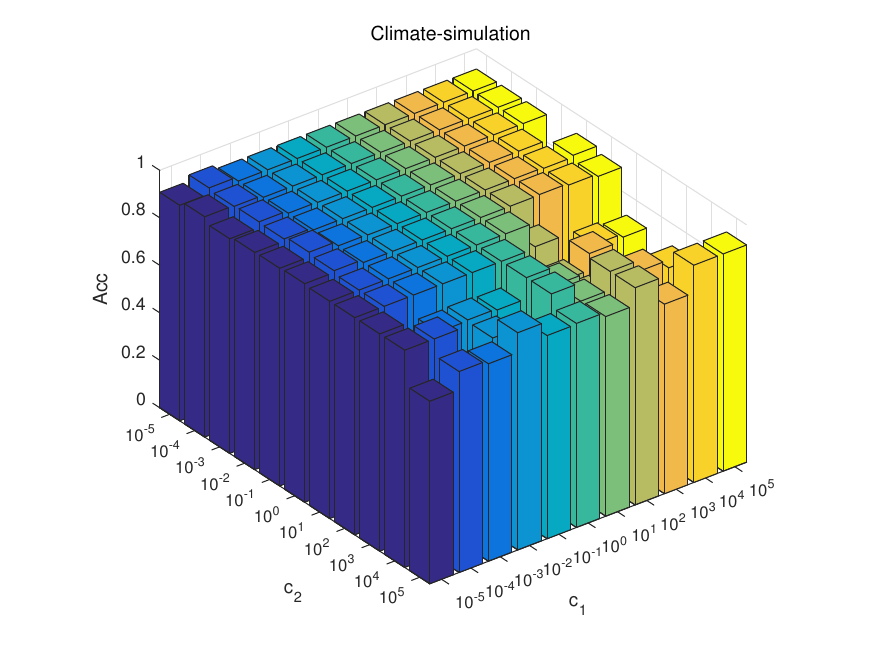}
\end{minipage}%
}%

\subfigure[Ionosphere]{
\begin{minipage}[t]{0.5\linewidth}
\centering
\includegraphics[width=6cm]{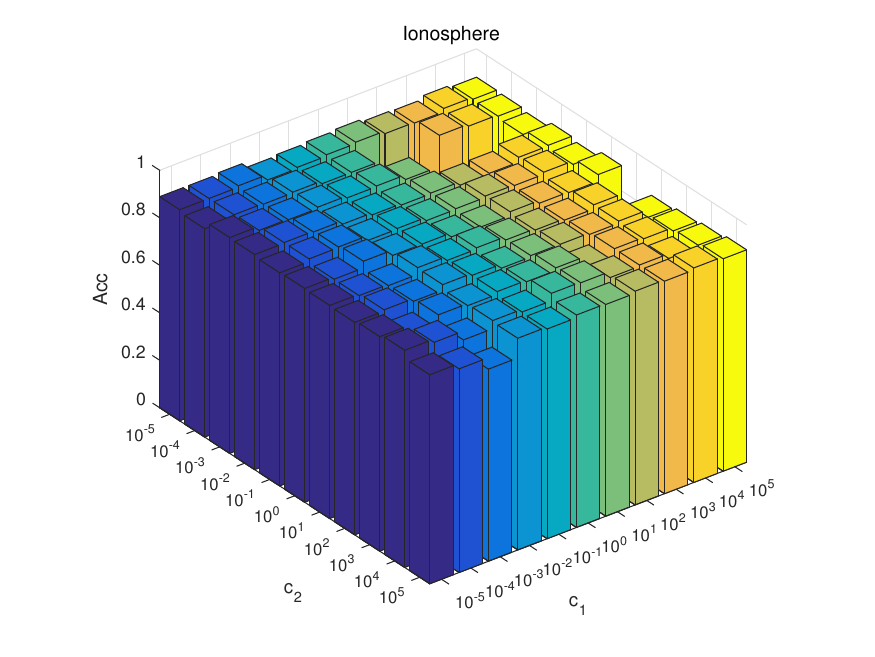}
\end{minipage}%
}%
\subfigure[Robotnavig]{
\begin{minipage}[t]{0.5\linewidth}
\centering
\includegraphics[width=6cm]{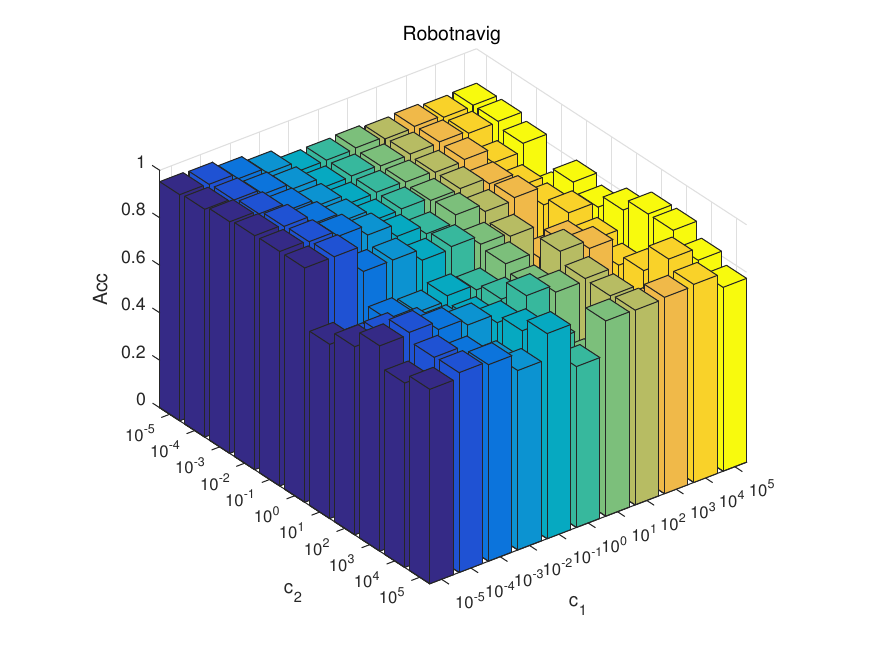}
\end{minipage}%
}%

\subfigure[Sonar]{
\begin{minipage}[t]{0.5\linewidth}
\centering
\includegraphics[width=6cm]{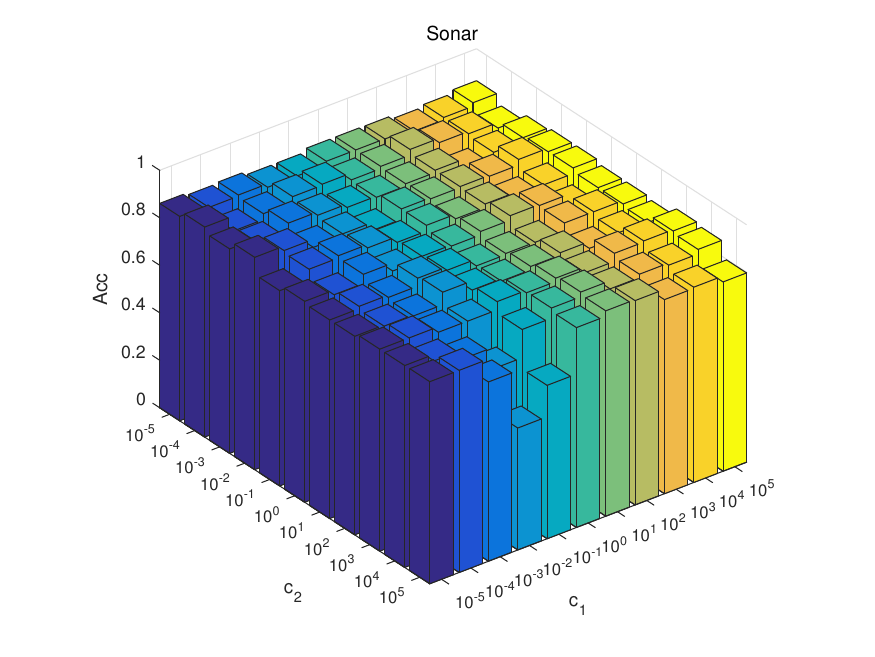}
\end{minipage}%
}%
\subfigure[Yeast(2 vs. 8)]{
\begin{minipage}[t]{0.5\linewidth}
\centering
\includegraphics[width=6cm]{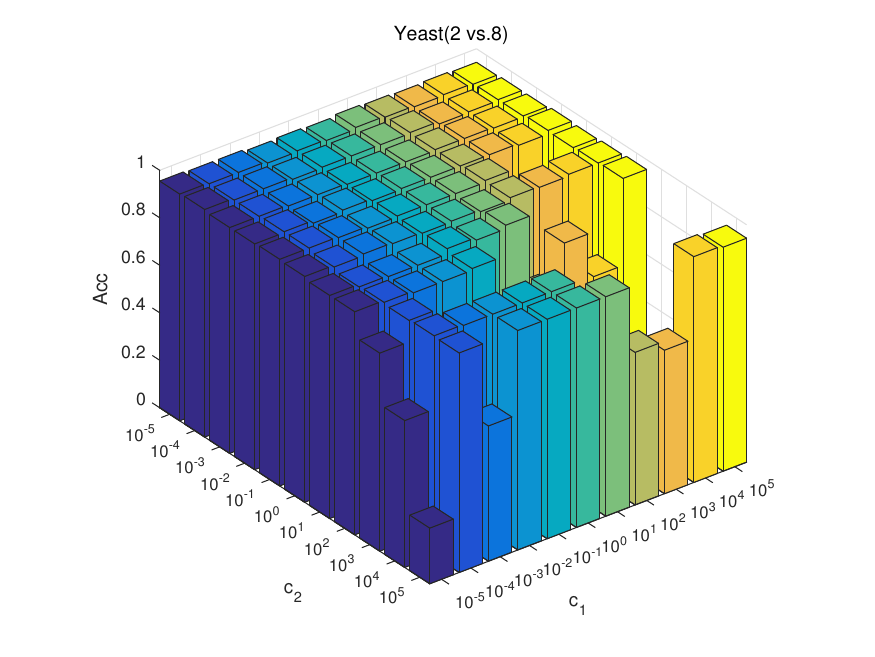}
\end{minipage}
}%
		\caption{Accuracy variation of C$L_1$QTSVM model with different hyperparameters.}\label{sanwei1}
\end{figure}

\subsection{Convergence analysis}
Figure \ref{shoulian1} shows the objective function value versus the number of iterations for our C$L_1$QTSVM method on the six benchmark datasets. Note that the objective function value is the sum of the objective function values of our optimization problem (\ref{CL1SLQTSVM1}) and (\ref{CL1SLQTSVM2}). In the convergence analysis experiments, 90$\%$ of the data points in each benchmark dataset are randomly selected as the training set, and the remaining 10$\%$ of the points are used as the test set. From the Figure \ref{shoulian1} we can observe that the objective function value decreases rapidly and converges within 10 iteration steps. Therefore, these convergence analysis results effectively validate the convergence and effectiveness of our proposed C$L_1$QTSVM model.

\begin{figure}[htbp]
\centering
\subfigure[Glass(0-1-5 vs .2)]{
\begin{minipage}[t]{0.5\linewidth}
\centering
\includegraphics[width=6cm]{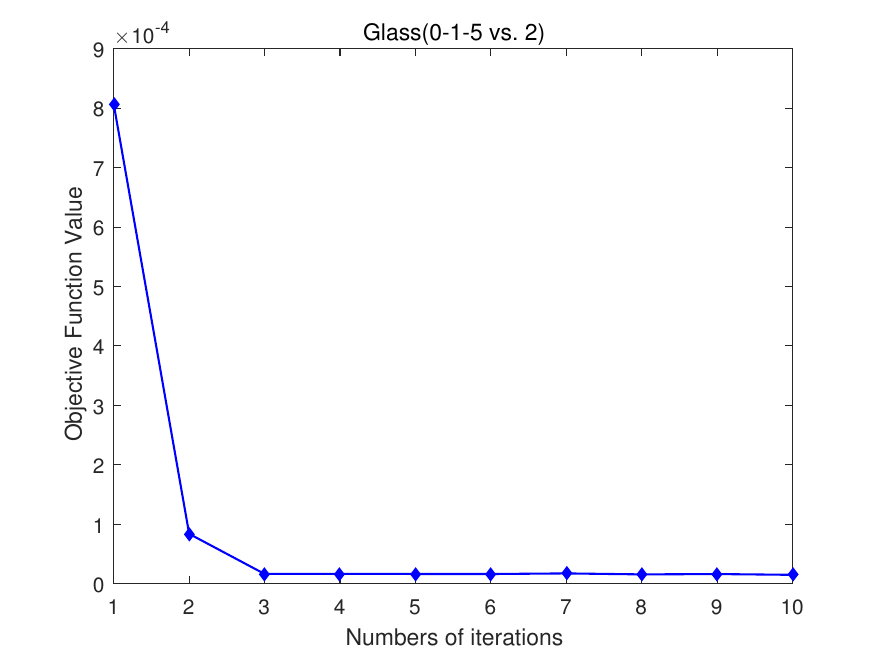}
\end{minipage}%
}%
\subfigure[Cylinder-bands]{
\begin{minipage}[t]{0.5\linewidth}
\centering
\includegraphics[width=6cm]{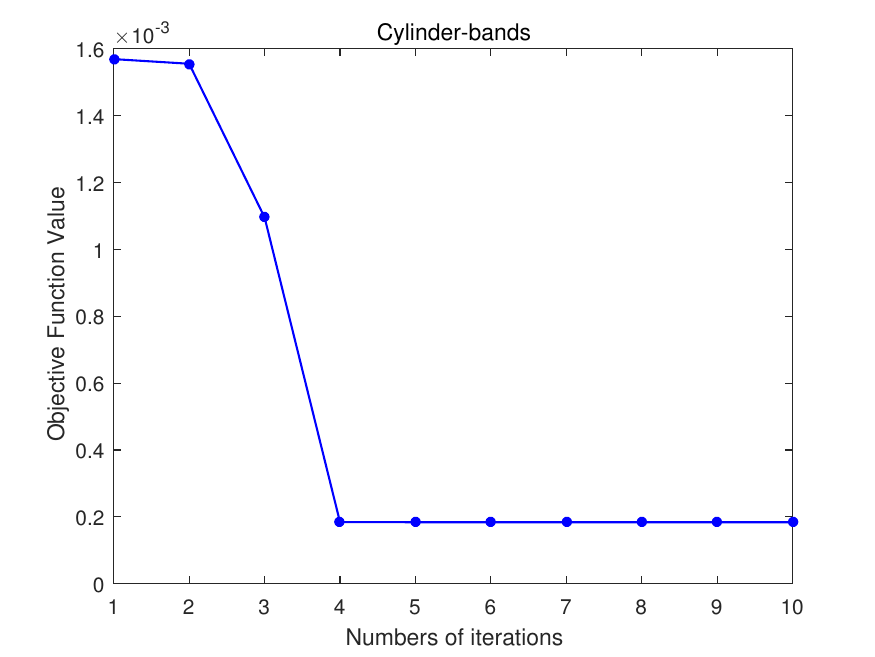}
\end{minipage}%
}%

\subfigure[Hepatitis]{
\begin{minipage}[t]{0.5\linewidth}
\centering
\includegraphics[width=6cm]{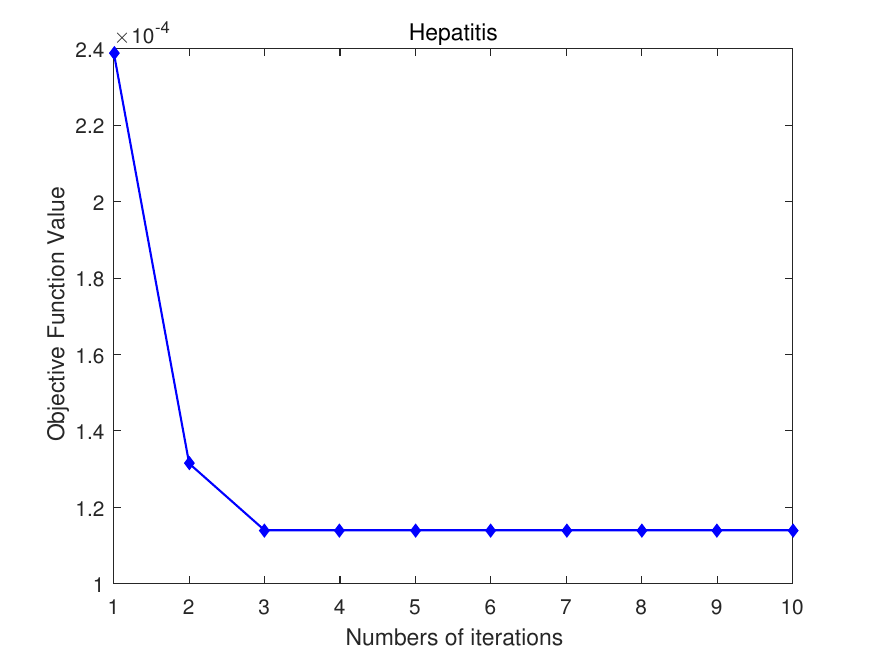}
\end{minipage}%
}%
\subfigure[Haberman]{
\begin{minipage}[t]{0.5\linewidth}
\centering
\includegraphics[width=6cm]{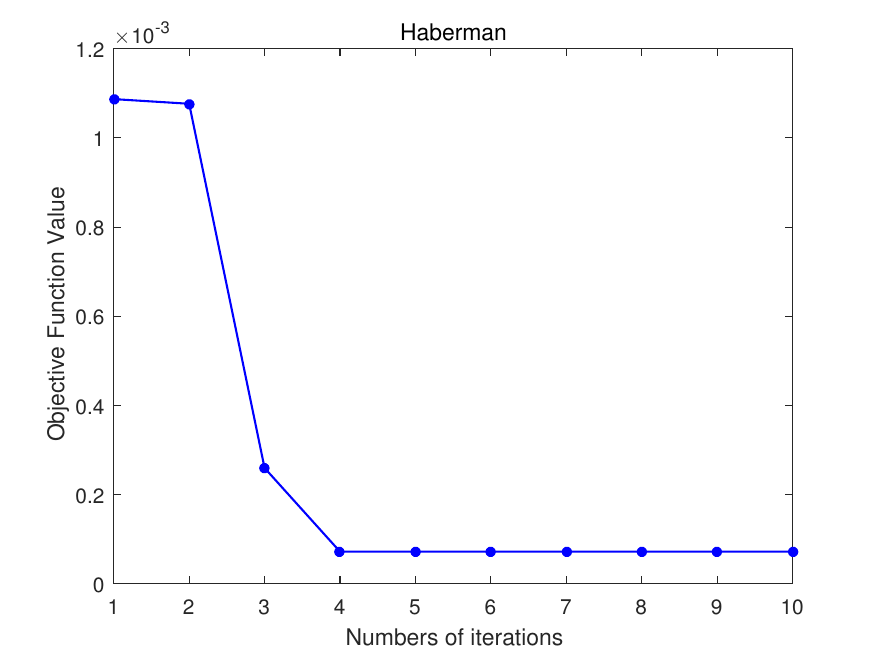}
\end{minipage}%
}%

\subfigure[Ionosphere]{
\begin{minipage}[t]{0.5\linewidth}
\centering
\includegraphics[width=6cm]{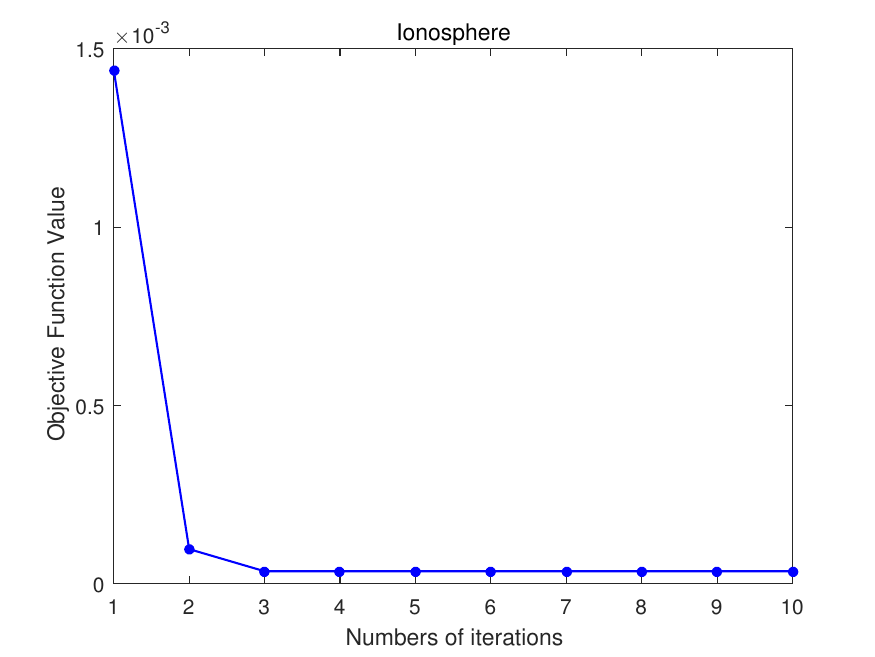}
\end{minipage}%
}%
\subfigure[Sonar]{
\begin{minipage}[t]{0.5\linewidth}
\centering
\includegraphics[width=6cm]{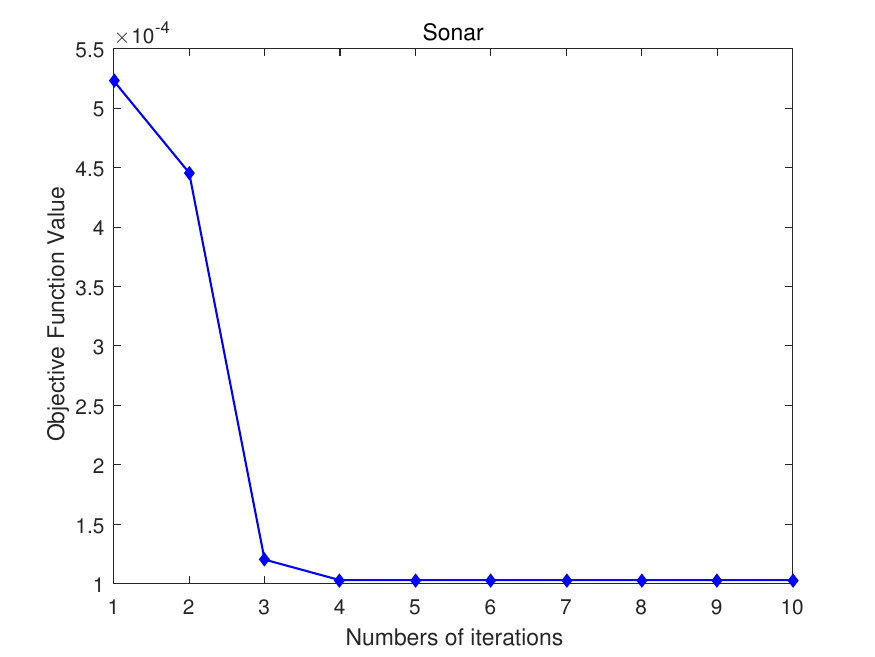}
\end{minipage}
}%
		\caption{Convergence of the C$L_1$QTSVM model on 6 benchmark datasets.}\label{shoulian1}
\end{figure}

\subsection{Non-parametric statistical test}
In this subsection, we use the Nemenyi post hoc test \citep{post} to further validate the differences between the C$L_1$QTSVM model and other comparative state-of-the-art methods. The principle of this test is that if the difference between the mean ranks of the two algorithms is greater than the calculated critical difference (CD), we consider the two algorithms to be significantly different. Otherwise, there is no significant difference between the two algorithms. In addition, the CD value is calculated as follows
\begin{equation}\label{448}
C D=q_{\alpha}(k) \times \sqrt{\frac{k(k+1)}{6 N}},
\end{equation}
where $\alpha=0.05$, and we have $q_{\alpha}=3.0310$.

As shown in Figure (\ref{posttest1}), there is no significant difference between the methods connected by the red line. Based on the experimental results on the synthetic dataset in Table \ref{rengongjieguo}, Figure (\ref{posttest1})(a) demonstrates the results of the corresponding Nemenyi post hoc test. From Figure (\ref{posttest1})(a), it can be seen that the average classification results of the proposed C$L_1$QTSVM model on the nine synthetic datasets are not significantly different from those of C$L_1$FRTBSVM, TSVM, and LSQTSVM; However, there is a significant difference between the classification performance of our C$L_1$QTSVM model and that of the remaining four comparison methods.  In addition, based on the experimental results on the benchmark dataset in Table \ref{jizhun1}, Figure (\ref{posttest1})(b) shows the corresponding Nemenyi post-test results. From Figure (\ref{posttest1})(b), it can be seen that the classification performance of the proposed C$L_1$QTSVM model on the 16 benchmark datasets is not significantly different from that of $\nu$-FRSQSSVM, C$L_1$FRTBSVM and C$L_{2,p}$-LSTSVM, but there is a significant difference between them and the remaining four methods.  From the above results, we can see that our C$L_1$QTSVM model is effective in dealing with synthetic data with outliers, benchmark data and nonlinear data.

\begin{figure}[htbp]
\centering
\subfigure[Synthetic datasets]{
\begin{minipage}[t]{0.5\linewidth}
\centering
\includegraphics[width=6cm]{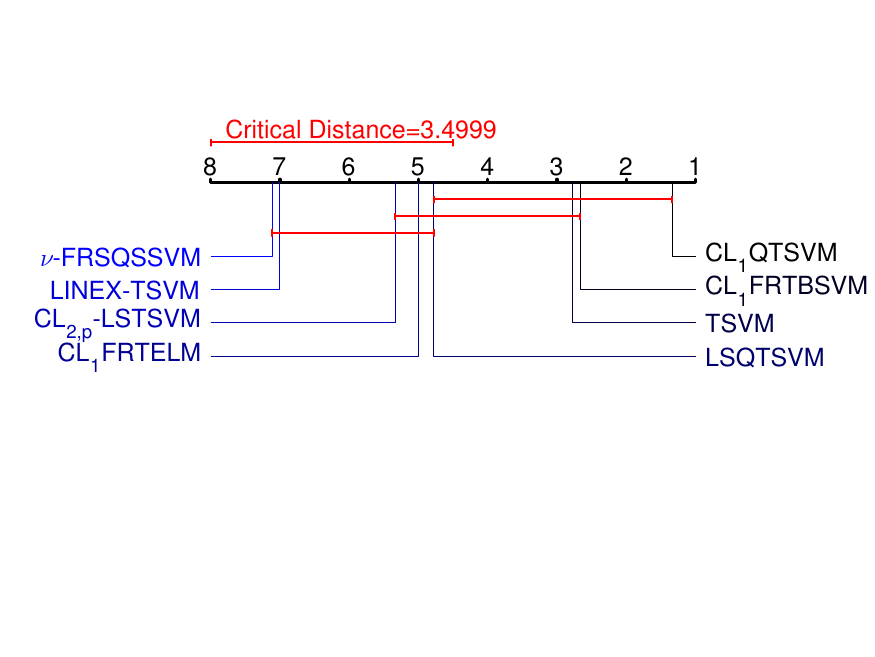}
\end{minipage}%
}%
\subfigure[Benchmark datasets]{
\begin{minipage}[t]{0.5\linewidth}
\centering
\includegraphics[width=6cm]{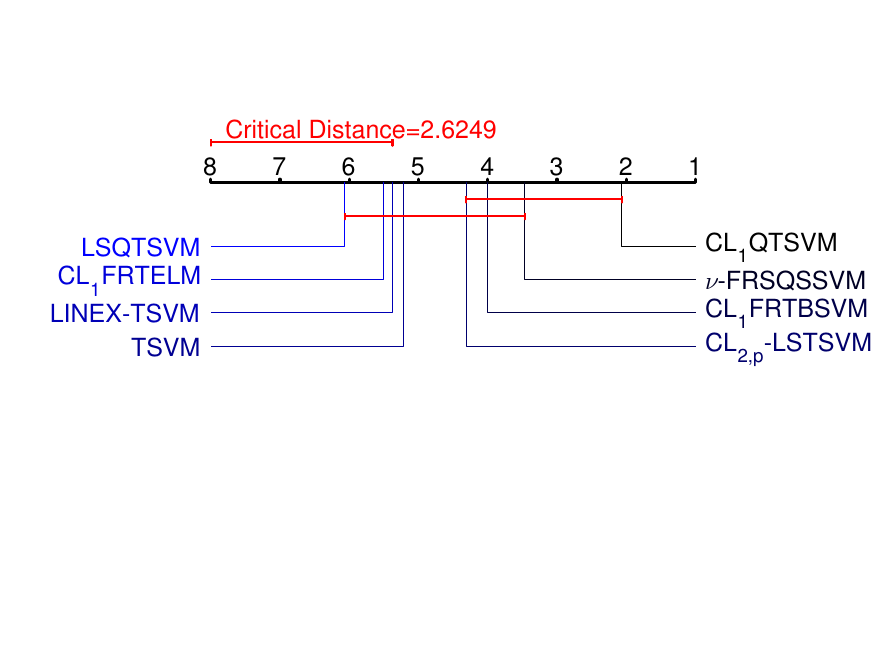}
\end{minipage}%
}
\centering
\caption{Nemenyi post hoc tests on synthetic and benchmark datasets.}\label{posttest1}
\end{figure}

\section{Conclusion}\label{6}
In this paper, a capped ${L}_1$-norm kernel-free quadratic support vector machine (C${L}_1$QTSVM) is proposed from the perspectives of distance norms and kernel-free techniques.  First, a non-smooth but bounded and symmetric capped ${L}_1$-norm distance measure is introduced to improve the robustness of the model to outliers. Meanwhile, the introduction of the ${L}_2$-norm regularization term improves the generalization ability of the model.  In addition, in order to avoid the time-consuming problem of selecting the appropriate kernel functions and kernel parameters for the traditional twin support vector machine model on some real datasets, this paper extends the model to a kernel-free version by utilizing the kernel-free technique, where only two quadratic hypersurfaces need to be found for the proposed C${L}_1$QTSVM model. For the established non-differentiable optimization problem, the original problem is equivalently transformed by reweighting technique, and an efficient iterative algorithm is constructed for the transformed optimization problem. In the theory section, the convergence, time complexity and existence of local optimal solutions of the designed iterative algorithm are analyzed in detail. 

In the numerical experiments, the experiments on synthetic dataset, benchmark dataset and image recognition dataset show that the C${L}_1$QTSVM model has certain advantages in terms of classification performance, robustness and computation time.  In the future, we can extend the C$L_1$QTSVM model to multi-view learning, multi-classification problems and clustering problems. Similarly, deep learning and functional data classification are also directions that can be considered.

\section*{Acknowledgements}
This work is supported by the National Natural Science Foundation of China (No.12061071) and Xinjiang Key Laboratory of Applied Mathematics (No.XJDX1401).

\section*{Conflict of interest}
The authors declare that they have no conflict of interest.

\bibliographystyle{elsarticle-harv}
\bibliography{CL1QTSVM}

\begin{thebibliography}{39}
\expandafter\ifx\csname natexlab\endcsname\relax\def\natexlab#1{#1}\fi
\expandafter\ifx\csname url\endcsname\relax
  \def\url#1{\texttt{#1}}\fi
\expandafter\ifx\csname urlprefix\endcsname\relax\def\urlprefix{URL }\fi

\bibitem[{Bai et~al.(2015)Bai, Han, Chen, and Yu}]{baiQSLSSVM}
Bai, Y., Han, X., Chen, T., Yu, H., 2015. Quadratic kernel-free least squares
  support vector machine for target diseases classification. Journal of
  Combinatorial Optimization 30, 850--870.

\bibitem[{Borah and Gupta(2020)}]{L2yichang}
Borah, P., Gupta, D., 2020. Functional iterative approaches for solving support
  vector classification problems based on generalized huber loss. Neural
  Computing and Applications 32~(13), 9245--9265.

\bibitem[{Brown et~al.(2000)Brown, Grundy, Lin, Cristianini, Sugnet, Furey,
  Ares~Jr, and Haussler}]{jiyin}
Brown, M.~P., Grundy, W.~N., Lin, D., Cristianini, N., Sugnet, C.~W., Furey,
  T.~S., Ares~Jr, M., Haussler, D., 2000. Knowledge-based analysis of
  microarray gene expression data by using support vector machines. Proceedings
  of the National Academy of Sciences 97~(1), 262--267.

\bibitem[{Chandra and Khemchandani(2007)}]{TSVM}
Chandra, S., Khemchandani, R., 2007. Twin support vector machines for pattern
  classification. IEEE Transactions on Pattern Analysis and Machine
  Intelligence 29~(5), 905--910.

\bibitem[{Cortes and Vapnik(1995)}]{hejishu}
Cortes, C., Vapnik, V., 1995. Support vector networks. Machine learning 20,
  273--297.

\bibitem[{Gao and Yu(2008)}]{shuangshijin1}
Gao, D.~Y., Yu, H., 2008. Multi-scale modelling and canonical dual finite
  element method in phase transitions of solids. International Journal of
  Solids and Structures 45~(13), 3660--3673.

\bibitem[{Gao et~al.(2019)Gao, Bai, and Zhan}]{gaoQSLSTSVM}
Gao, Q.-Q., Bai, Y.-Q., Zhan, Y.-R., 2019. Quadratic kernel-free least square
  twin support vector machine for binary classification problems. Journal of
  the Operations Research Society of China 7, 539--559.

\bibitem[{Gao et~al.(2021{\natexlab{a}})Gao, Fang, Gao, Luo, and
  Medhin}]{gao-reg-LSDWPTSVM}
Gao, Z., Fang, S.-C., Gao, X., Luo, J., Medhin, N., 2021{\natexlab{a}}. A novel
  kernel-free least squares twin support vector machine for fast and accurate
  multi-class classification. Knowledge-Based Systems 226, 107123.

\bibitem[{Gao et~al.(2021{\natexlab{b}})Gao, Fang, Luo, and
  Medhin}]{gao-DWPSVM}
Gao, Z., Fang, S.-C., Luo, J., Medhin, N., 2021{\natexlab{b}}. A kernel-free
  double well potential support vector machine with applications. European
  Journal of Operational Research 290~(1), 248--262.

\bibitem[{Gao et~al.(2022)Gao, Wang, Huang, Luo, and Tang}]{gaonu-FRSQSSVM}
Gao, Z., Wang, Y., Huang, M., Luo, J., Tang, S., 2022. A kernel-free fuzzy
  reduced quadratic surface $\nu$-support vector machine with applications.
  Applied Soft Computing 127, 109390.

\bibitem[{Garc{\'\i}a et~al.(2010)Garc{\'\i}a, Fern{\'a}ndez, Luengo, and
  Herrera}]{post}
Garc{\'\i}a, S., Fern{\'a}ndez, A., Luengo, J., Herrera, F., 2010. Advanced
  nonparametric tests for multiple comparisons in the design of experiments in
  computational intelligence and data mining: Experimental analysis of power.
  Information Sciences 180~(10), 2044--2064.

\bibitem[{Goh et~al.(2005)Goh, Chang, and Li}]{yingyong4}
Goh, K.-S., Chang, E.-Y., Li, B., 2005. Using one-class and two-class svms for
  multiclass image annotation. IEEE transactions on Knowledge and Data
  Engineering 17~(10), 1333--1346.

\bibitem[{Gupta et~al.(2020)Gupta, Hazarika, and Berlin}]{Huber-huigui-ELM}
Gupta, D., Hazarika, B.~B., Berlin, M., 2020. Robust regularized extreme
  learning machine with asymmetric huber loss function. Neural Computing and
  Applications, 1--28.

\bibitem[{Hearst et~al.(1998)Hearst, Dumais, Osuna, Platt, and Scholkopf}]{SVM}
Hearst, M.~A., Dumais, S.~T., Osuna, E., Platt, J., Scholkopf, B., 1998.
  Support vector machines. IEEE Intelligent Systems and their applications
  13~(4), 18--28.

\bibitem[{Huang et~al.(2019)Huang, Shao, Zhang, Zhao, and Teng}]{shao}
Huang, L.-W., Shao, Y.-H., Zhang, J., Zhao, Y.-T., Teng, J.-Y., 2019. Robust
  rescaled hinge loss twin support vector machine for imbalanced noisy
  classification. IEEE Access 7, 65390--65404.

\bibitem[{Kumar and Gopal(2009)}]{LSTSVM}
Kumar, M.~A., Gopal, M., 2009. Least squares twin support vector machines for
  pattern classification. Expert systems with applications 36~(4), 7535--7543.

\bibitem[{Luo et~al.(2016)Luo, Fang, Deng, and Guo}]{luoSQSSVM}
Luo, J., Fang, S.-C., Deng, Z., Guo, X., 2016. Soft quadratic surface support
  vector machine for binary classification. Asia-Pacific Journal of Operational
  Research 33~(06), 1650046.

\bibitem[{Ma(2020)}]{L1TELM}
Ma, J., 2020. Capped l 1-norm distance metric-based fast robust twin extreme
  learning machine. Applied Intelligence 50~(11), 3775--3787.

\bibitem[{Ma et~al.(2020)Ma, Yang, and Sun}]{FRTBSVM}
Ma, J., Yang, L., Sun, Q., 2020. Capped l1-norm distance metric-based fast
  robust twin bounded support vector machine. Neurocomputing 412, 295--311.

\bibitem[{Ma et~al.(2018)Ma, Cheng, Shang, and Liu}]{guzhang}
Ma, S., Cheng, B., Shang, Z., Liu, G., 2018. Scattering transform and lsptsvm
  based fault diagnosis of rotating machinery. Mechanical Systems and Signal
  Processing 104, 155--170.

\bibitem[{Moosaei et~al.(2023)Moosaei, Ganaie, Hlad{\'\i}k, and
  Tanveer}]{TSVM2}
Moosaei, H., Ganaie, M., Hlad{\'\i}k, M., Tanveer, M., 2023. Inverse free
  reduced universum twin support vector machine for imbalanced data
  classification. Neural Networks 157, 125--135.

\bibitem[{Nie et~al.(2014)Nie, Huang, Wang, and Huang}]{32}
Nie, F., Huang, Y., Wang, X., Huang, H., 2014. New primal svm solver with
  linear computational cost for big data classifications. in: International
  Conference on Machine Learning (II-505).

\bibitem[{Noble(2004)}]{yingyong3}
Noble, W.-S., 2004. Support vector machine applications in computational
  biology. Kernel Methods in Computational Biology 71, 92.

\bibitem[{Shao et~al.(2011)Shao, Zhang, Wang, and Deng}]{zhengzexiang}
Shao, Y.-H., Zhang, C.-H., Wang, X.-B., Deng, N.-Y., 2011. Improvements on twin
  support vector machines. IEEE Transactions on Neural Networks 22~(6),
  962--968.

\bibitem[{Si et~al.(2023)Si, Yang, and Ye}]{SLTSVM}
Si, Q., Yang, Z., Ye, J., 2023. Symmetric linex loss twin support vector
  machine for robust classification and its fast iterative algorithm. Neural
  Networks.

\bibitem[{Tanveer et~al.(2019)Tanveer, Sharma, and Suganthan}]{PTSVM}
Tanveer, M., Sharma, A., Suganthan, P.-N., 2019. General twin support vector
  machine with pinball loss function. Information Sciences 494, 311--327.

\bibitem[{Wang et~al.(2019)Wang, Ye, Luo, Ye, and Fu}]{cappedL1TSVM}
Wang, C., Ye, Q., Luo, P., Ye, N., Fu, L., 2019. Robust capped l1-norm twin
  support vector machine. Neural Networks 114, 47--59.

\bibitem[{Wang et~al.(2023)Wang, Yu, and Ma}]{CL2p1}
Wang, H., Yu, G., Ma, J., 2023. Capped l 2, p-norm metric based on robust twin
  support vector machine with welsch loss. Symmetry 15~(5), 1076.

\bibitem[{Wang et~al.(2022)Wang, Yu, and \&~Ma}]{LINEX-TSVM}
Wang, Y.-F., Yu, G.-L., \&~Ma, J., 2022. Capped linex metric twin support
  vector machine for robust classification. Sensors 22~(17), 6583.

\bibitem[{Wu et~al.(2017)Wu, Liu, Gao, Kong, and Feng}]{27}
Wu, M.-J., Liu, J.-X., Gao, Y.-L., Kong, X.-Z., Feng, C.-M., 2017. Feature
  selection and clustering via robust graph-laplacian pca based on capped l
  1-norm. In: 2017 IEEE international conference on Bioinformatics and
  Biomedicine (BIBM). pp. 1741--1745.

\bibitem[{Xia et~al.(2014)Xia, Sheu, Fang, and Xing}]{shuangshijin2}
Xia, Y., Sheu, R.-L., Fang, S.-C., Xing, W., 2014. Double well potential
  function and its optimization in the n-dimensional real space-part ii.
  Journal of Industrial and Management Optimization.

\bibitem[{Xiang et~al.(2012)Xiang, Nie, Meng, Pan, and Zhang}]{30}
Xiang, S., Nie, F., Meng, G., Pan, C., Zhang, C., 2012. Discriminative least
  squares regression for multiclass classification and feature selection. IEEE
  transactions on neural networks and learning systems 23~(11), 1738--1754.

\bibitem[{Xie et~al.(2023{\natexlab{a}})Xie, Li, and Sun}]{xieshenduTSVM}
Xie, X., Li, Y., Sun, S., 2023{\natexlab{a}}. Deep multi-view multiclass twin
  support vector machines. Information Fusion 91, 80--92.

\bibitem[{Xie et~al.(2023{\natexlab{b}})Xie, Sun, Qian, Guo, Zhang, Ye, and
  Wang}]{LpLSTSVM}
Xie, X., Sun, F., Qian, J., Guo, L., Zhang, R., Ye, X., Wang, Z.,
  2023{\natexlab{b}}. Laplacian lp norm least squares twin support vector
  machine. Pattern Recognition 136, 109192.

\bibitem[{Xie and Xiong(2022)}]{duoshijiao1}
Xie, X., Xiong, Y., 2022. Generalized multi-view learning based on generalized
  eigenvalues proximal support vector machines. Expert Systems with
  Applications 194, 116491.

\bibitem[{Yan et~al.(2018)Yan, Ye, Zhang, Yu, Yuan, Xu, and Fu}]{24}
Yan, H., Ye, Q., Zhang, T., Yu, D.-J., Yuan, X., Xu, Y., Fu, L., 2018. Least
  squares twin bounded support vector machines based on l1-norm distance metric
  for classification. Pattern recognition 74, 434--447.

\bibitem[{Yang et~al.(2023)Yang, Xue, Ma, and Chang}]{L1TPELM}
Yang, Y., Xue, Z., Ma, J., Chang, X., 2023. Robust projection twin extreme
  learning machines with capped l1-norm distance metric. Neurocomputing 517,
  229--242.

\bibitem[{Yuan and Yang(2021)}]{CL2p}
Yuan, C., Yang, L., 2021. Capped l2, p-norm metric based robust least squares
  twin support vector machine for pattern classification. Neural Networks 142,
  457--478.

\bibitem[{Zheng et~al.(2021)Zheng, Zhang, and Yan}]{Hingeyichang}
Zheng, X.-H., Zhang, L., Yan, L.-L., 2021. Ctsvm: a robust twin support vector
  machine with correntropy-induced loss function for binary classification
  problems. Information Sciences 559, 22--45.

\end{thebibliography}
\end{document}